\documentclass{article}

\usepackage{amsmath}
\usepackage{amsthm}
\usepackage{mathtools}
\usepackage{bbm}
\usepackage{mathrsfs}
\usepackage{xcolor}
\usepackage{authblk}
\usepackage{bbold}
\usepackage{fancyhdr}

\usepackage[T1]{fontenc}
\usepackage{amsfonts}
\usepackage{vmargin}
\usepackage{setspace}
\usepackage{esint}

\fancypagestyle{plain}{%
  \fancyhf{}%
}

\newcommand{\one}{\mathbbm{1}}
\newcommand{\dd}{\mathrm{d}}
\newcommand{\R}{\mathbb{R}}
\newcommand{\Z}{\mathbb{Z}}

\newcommand{\Ibog}{I_d^{\rm{Bog}}}
\newcommand{\f}{\lambda_d^{\rm{LHY}}}

\usepackage{cleveref}

\usepackage{mathrsfs, enumerate}

\numberwithin{equation}{section}
\usepackage[explicit]{titlesec}
\titleformat{\section}{\centering\Large\bfseries}{\thesection \ --}{0.7em}{\Large\bfseries #1}
\titleformat{\subsection}{\centering\large\bfseries}{\thesubsection \ --}{0.4em}{\large\bfseries #1}
\titleformat{\subsubsection}{\centering\bfseries}{\thesubsubsection \ --}{0.4em}{\bfseries #1}

\theoremstyle{definition} \newtheorem{definition}{Definition}[section]
\theoremstyle{plain} \newtheorem{theorem}[definition]{Theorem}
\theoremstyle{plain} 
\theoremstyle{plain} \newtheorem{proposition}[definition]{Proposition}
\theoremstyle{plain} \newtheorem{lemma}[definition]{Lemma}
\theoremstyle{plain} 
\theoremstyle{plain} \newtheorem{remark}[definition]{Remark}
\theoremstyle{definition}

\begin{document}

\markboth{Fournais et al.}{Lower bounds on the energy of the Bose gas}

\title{Lower bounds on the energy of the Bose gas}

\author[1]{Søren Fournais\thanks{fournais@math.au.dk}}

\author[2]{Theotime Girardot\thanks{theotime.girardot@math.au.dk}}

\author[3]{Lukas Junge\thanks{junge@math.au.dk}}

\author[4]{Leo Morin\thanks{leo.morin@math.au.dk}}

\author[5]{Marco Olivieri\thanks{marco.olivieri@math.au.dk}}

\affil[1,2,3,4,5]{Department of Mathematics, Aarhus University, Ny Munkegade 118,  DK-8000 Aarhus C, Denmark}

\maketitle

\begin{abstract}
We present an overview of the approach to establish a lower bound to the ground state energy for the dilute, interacting Bose gas in a periodic box. In this paper the size of the box is larger than the Gross-Pitaevski length scale.
The presentation includes both the 2 and 3 dimensional cases, and catches the second order correction, i.e. the Lee-Huang-Yang term. The calculation on a box of this length scale is the main step to calculate the energy in the thermodynamic limit. However, the periodic boundary condition simplifies many steps of the argument considerably compared to the localized problem coming from the thermodynamic case.
\end{abstract}

\section{Introduction and Main Results}
\subsection{Introduction}
The understanding of the ground state of a Bose gas is of major interest in many-body quantum theory, especially since the first experimental observation of Bose-Einstein condensates \cite{AndEnsMatWieCor_95}. It is a very challenging problem to find properties of this ground state, and the mathematical proof of condensation in the thermodynamic limit is still out of reach. In this paper, we focus on the asymptotic behaviour of the ground state {\it energy} in the dilute limit, both in dimensions $d=2$ and $3$. 

To state the results, we consider a gas of $N$ bosons in a box $\Omega$, in the thermodynamic limit $\lvert \Omega \rvert \rightarrow \infty$, with fixed density $\rho = N / \lvert \Omega\rvert$. The first terms of the expansion of the ground state energy density of such a gas depend only on the scattering length $a$ of the inter-particle potential (as defined in Section \ref{Sec:scatt} below) and the density $\rho$. In the $3$-dimensional case, the ground state energy density has the following expansion in dilute limit $\rho a^3 \rightarrow 0$,
\begin{align}\label{eq:LHY3D}
	e^{\rm{3D}}(\rho) = 4\pi \rho^2 a  \Big( 1 + \frac{128}{15\sqrt{\pi}} \sqrt{\rho a^3} \Big) + o\big(  \rho^2 a\sqrt{\rho a^3} \big).
\end{align}
The leading term of this asymptotic formula was first derived in \cite{Lenz}, and the second term, the Lee-Huang-Yang term, was given in \cite{Bog,LHY}. 
Mathematical proofs of the leading order term were given in \cite{dyson} for the upper bound and in \cite{LY} for the lower bound.
The first upper bound to LHY precision was given in \cite{ESY}, and the correct constant in \cite{YY} with recent improvements in \cite{BCS}, for sufficiently regular potentials.
The matching lower bounds were given in \cite{FS, FS2} including the crucial case of hard core potentials.
The upper bound in the case of potentials with large $L^1$-norm, such as the hard core interactions, is still an open problem. However, the reader may find recent improvements in \cite{BCOPS}.

In the $2$-dimensional case, the asymptotic formula is
\begin{align}\label{eq:MainHC}
	e^{\rm{2D}}(\rho)= 4\pi \rho^2 \delta \Big(1 + \Big[2\Gamma + \frac{1}{2} + \log(\pi) \Big] \delta	\Big) +o(\rho^2 \delta^2),
\end{align}
where $\Gamma \simeq 0.57$ is the Euler-Mascheroni constant and $\delta$ is a small logarithmic parameter given by
\begin{align}\label{eq:def_delta}
	\delta := \frac{1}{|\log(\rho a^2 |\log(\rho a^2)|^{-1})|}.
\end{align}
This formula was first given in \cite{PhysRevE.64.027105,HINES197812,PhysRevLett.102.180404,Yang_2008}, and the leading order was first proven in \cite{MR1827922}. Both upper and lower bounds to second order precision were recently proved in \cite{2DLHY}, and they include the case of hard core interactions.
We refer to \cite{BasCarCen_22,Sch_22} for overview articles.
Similar expansions for Bose gases in $2D$ were obtained, in the Gross-Pitaevskii regime in \cite{CCS2D} or in different regimes, see \cite{MR3986933}.

The case of interacting Fermi gases is equally interesting and has seen major progress in recent years, see for instance
\cite{LSSfermions, Hainzl1, Hainzl2, Hainzl3, benedikter_optimal_2020, benedikter_correlation_2021, falconi_dilute_2021, GiacoFermi, GiacoFermi2, LauritsenFermi, LauritsenFermi2}.

The purpose of the present paper is to explain the proof of lower bounds in \cite{FS} and \cite{2DLHY} for the $3D$ and the $2D$ case, respectively, which are similar in many aspects. The very first step, both in $2D$ and $3D$, is to reduce the problem to lengthscales $\ell$ which are much smaller than the thermodynamic length $L$ but larger than the Gross-Pitaevski length scale. This localization procedure is now quite standard \cite{brietzke_second-order_2020}, but gives rise to technical complications. Mainly, the kinetic energy is inconveniently modified, including localization functions which affect the algebra of calculations and require more involved estimates. For this reason, we decide here to directly consider a gas of bosons on a periodic box of the right $\rho$-dependent length scale and to carry out all the analysis in this setting omitting the localization step. Since many terms are simpler and many errors vanish, this should help the interested reader understand the general strategy of lower bounds for Bose gases.\\

Before introducing the energy and the associated result we need to recall some basic facts about the scattering equation.
\subsection{Scattering length}\label{Sec:scatt}

An important difference between 2 and 3 dimensions concerns the properties of scattering solutions, which can be found in \cite[Appendix A]{greenbook}. We recall here the main definitions, and fix notations.

In this paper we will only consider radial, compactly supported and positive potentials $v:\mathbb{R}^d\rightarrow [0,\infty]$, with $R>0$ such that $\textrm{supp} (v) \subseteq B^d(0,R)$, where we denote by $B^d(y,r)$ the ball of radius $r$ centered in $y$ in $\mathbb{R}^d$.

Let us consider the minimization problem, for an arbitrary $\widetilde{R} > R$,
\begin{equation}\label{variational definition of the scattering length}	\mathrm{E}_{d}(v,\widetilde{R})=\inf_{\varphi}\int_{B^{d}(0,\widetilde{R})}\Big(\vert \nabla \varphi \vert^2+\frac{1}{2}v\varphi^2 \Big)\dd x,
\end{equation}
where the infimum is taken over $\varphi \in H^1(B^{d}(0,\widetilde{R}))$ such that $ \varphi \vert_{\partial B^{d}(0,\widetilde{R})}=1$. 
We define the scattering length $a=a(v)$ by
\begin{equation}
	\mathrm{E}_{2}(v,\widetilde{R})=\frac{2\pi}{\log (\frac{\widetilde R}{a})},\qquad \text{and}\qquad \mathrm{E}_{3}(v,\widetilde{R})=\frac{4\pi a}{1-a/\widetilde R}.
\end{equation}
It is a well-known result that $a$ is independent of $\widetilde{R}>R$.
The associated minimizers are of the form
\begin{equation}
	\varphi_{\mathbb{R}^d}=
	\begin{dcases}
		\frac{1}{\log( \tilde R / a)}\varphi^{0}_{\mathbb{R}^d}, & \text{ if } d=2,\\
		\frac{1}{1-a/\widetilde R} \varphi^{0}_{\R^d}, & \text{ if } d=3,
	\end{dcases}
\end{equation}
where $\varphi^0_{\mathbb R^d}$ solves the scattering equation
\begin{equation} \label{scattering_equation}
	-\Delta \varphi^{0}_{\mathbb{R}^d}+\frac{1}{2}v\varphi_{\mathbb{R}^d}^{0} =0,
\end{equation}
in a distributional sense. The solution is such that, for $\lvert x \rvert \geq R$, we have the explicit form
\begin{align}
	\varphi^{0}_{\mathbb{R}^2}(x) =\log\Big(\frac{\lvert x \rvert}{a}\Big), \quad \text{and} \quad
	\varphi^{0}_{\mathbb{R}^3}(x) :=1-\frac{a}{\lvert x \rvert}.
\end{align}
If $d=3$, we choose $\widetilde R = \infty$ so that $\varphi^0_{\R^3} = \varphi_{\R^3}$. The logarithm in the 2D-scattering solution is clearly unbounded for large values of $\lvert x \rvert$. This is a major difference to the 3D behaviour.
Therefore the length $\widetilde{R}$ is of much greater importance. In this paper, when $d=2$, we choose
\begin{equation}\label{eq:delta}
	\widetilde{R}=ae^{\frac{1}{2\delta}}, \qquad  \text{i.e.}  \qquad \delta=\frac{1}{2}\log\Big(\frac{\widetilde{R}}{a}\Big)^{-1},
\end{equation}
so that \begin{equation}
	\varphi_{\mathbb{R}^2}:=2\delta \varphi^{0}_{\mathbb{R}^2}
\end{equation}
is then normalized to $1$ at distance $\widetilde R$, with $\delta$ given in \eqref{eq:def_delta}.

\subsection{Main result}
We consider $N$ interacting bosons on the torus of unit cell $\Lambda = \big[ - \frac \ell 2, \frac \ell 2 \big]^d$. We define the associated Hamiltonian with periodic boundary conditions
\begin{equation}\label{def:HN}
	\mathcal H_N = \sum_{j=1}^N - \Delta_j + \sum_{1 \leq i < j \leq N} v(x_i-x_j), 
\end{equation}
acting on the space of symmetric square integrable functions $L^2_{\rm{sym}}(\Lambda^N)$, where $-\Delta$ is the periodic Laplacian on $\Lambda$ and the potential depends on $ (x_i - x_j)^{*}$, the distance between particle $i$ and $j$ on the torus. More precisely, we define $x^{*}\in \R$ by
\begin{equation}
	x^{*}=\min_{z\in \mathbb{Z}^{d}}\vert x-z\ell\vert ,
\end{equation} 
and 
\begin{equation}
	v(x)=v_{\R^{d}}(x^{*}), \quad \text{ with } \quad v_{\R^{d}}:\R^{d}\to \R_{+}.
\end{equation}

We assume $v_{\R^{d}}$ to be a \textit{positive, radially symmetric interaction with support in the ball of radius} $R\leq \ell/4$. This condition on the support will be made precise later, all we need for now is the support of $v$ to fit in the box. 
We have here committed a mild abuse of notation using that $v_{\mathbb{R}^d}$ is radially symmetric.  Using the positivity of the potential it is standard that $\mathcal{H}_{N}$ defines a self-adjoint operator.
If $\varphi_{\R^{d}}$ is the scattering solution associated to $v_{\R^{d}}$, we define
\begin{equation}
	\omega_{\R^{d}} := 1- \varphi_{\R^{d}}, \qquad g_{\R^{d}} := v_{\R^{d}}(1-\omega_{\R^{d}})=v_{\R^{d}}\varphi_{\R^{d}},
\end{equation}
and their periodic versions
\begin{equation}
	\omega(x) :=	\omega_{\R^{d}}(x^{*}), \qquad g(x) :=g(x^{*}), \qquad x\in \Lambda.
\end{equation}
Note that we dropped the dependence on $d$ in the notation. 
The function $g$ has a specific role in the analysis, and its Fourier transform satisfies, through a manipulation of the scattering equation \eqref{scattering_equation}, the relation
\begin{equation}\label{def:g00}
	\widehat g(0) = \begin{cases}
		8 \pi \delta, &\text{if } d=2,\\
		8 \pi a, &\text{if } d=3.
	\end{cases}
\end{equation}
Notice that since $R\leq \ell/4$, we have that the Fourier transforms and Fourier coefficients agree at zero, $\mathrm{i.e.}$ $\widehat{g}(0)=\widehat{g}_{\R^{d}}(0)$.
We scale the system in the following way: for a given density $\rho$ we define
\begin{equation}
	\ell:=\frac{K_{\ell}}{\sqrt{\rho \widehat{g}(0)}}
\end{equation}
where $K_{\ell}\gg 1$ is a large $\rho$-dependent parameter chosen in \eqref{condition.on.Kl}. This scaling has to be understood under the dilute regime assumption, that is $\rho a^{d}\leq C^{-1}$ for a large enough constant $C$. The regime $K_{\ell}=1$ corresponds to the well-known Gross-Pitaevskii regime. In this paper, the particular choice $K_{\ell}\gg 1$ is needed to control the errors obtained at the different steps of the proof, as the c-number substitution of Section \ref{Sec:cn} and to go from sums to integrals at a negligible cost, in particular to get the correct LHY constant.

The number $N$ of particles in the box is defined through 
\begin{equation*}
N= \rho \ell^d.
\end{equation*}

We can observe using \eqref{scattering_equation} that the Fourier transform $\widehat{g \omega}(0)$ can be written by means of an auxiliary function
\begin{equation}\label{defG}
	\widehat{g\omega}(0) = \frac{1}{(2\pi)^d} \int_{\mathbb{R}^d} G_d(k) \dd k, \qquad G_d(k) = \frac{\widehat g_{\R^{d}}(k)^2 - \widehat g_{\R^{d}}(0) ^2\one_{d}(\ell_{\delta} k)}{2 k^2},
\end{equation}
where we introduced the cut-off
\begin{equation}
	\one_{d} (t) := 
		\delta_{d,2}\one_{\{|t|\leq 1\}} (t), 
\end{equation}
with $\delta_{i,j}$ being the Kronecker delta, to deal with the 2D case where the Fourier transform of a logarithm involves a renormalization around zero. This renormalization is done at the scale
\begin{equation}
	\ell_{\delta}= \frac{a}{2} e^{\frac{1}{2\delta}}e^{\Gamma} =\frac{1}{2\sqrt{\rho\delta}}e^{\Gamma} (1 + o(1)),
\end{equation}
where we recall that $\Gamma$ is the Euler-Mascheroni constant.
We define the \textit{Lee-Huang-Yang energy} in dimension $d$ as
\begin{equation}\label{def:LHY-2ndenergy}
	E^{\text{LHY}}_d(\rho, \Lambda) := \frac{\rho^2}{2}|\Lambda|\widehat{g}(0)  \f I_d^{\text{Bog}},
\end{equation}
where  
\begin{equation}
	\f = \begin{dcases}  
		\sqrt{\rho a^3}, \quad &\text{if } d=3,\\
		\delta, \quad &\text{if } d=2, 
	\end{dcases}
\end{equation}
is the Lee-Huang-Yang correction order, and 
\begin{equation}\label{eq:ibog}
	I^{\text{Bog}}_{d}:=\Big(\frac{2}{\pi}\Big)^{d/2} \int_{\mathbb{R}^d}  \sqrt{(t^2 + 1)^2 - 1} - t^2 - 1 + \frac{ 1}{2 t^2}\big(1 + \one_d(\sqrt{2\pi} e^{\Gamma} t) \big) \,\dd t,
\end{equation}
is the Bogoliubov integral of dimension $d$.

We also define the LHY error in dimension $d$ denoted $o^{\mathrm{LHY}}_{d}$ as a quantity of smaller order than the LHY precision in term of the small parameter of the dilute regime $\rho a^{d}$. For any error term $\mathcal{E}$ we write $\mathcal{E}=  o^{\mathrm{LHY}}_{d} $ if there exist constants $C >0$ and $\eta >0$ such that
\begin{equation}
	\vert \mathcal{E}\vert \leq \begin{cases}
		C\rho^{2}|\Lambda| \delta^{2+\eta}, & \text{if } d=2, \\
		C \rho^{2} |\Lambda|a\big(\rho a^{3}\big)^{\frac 1 2 + \eta}, &\text{if } d=3.
	\end{cases}
\end{equation}

Let us recall the expressions of $\mathcal{H}_N$ and $\Lambda$ below, for reader's convenience:
\begin{align*}
&\mathcal{H}_N = \sum_{j=1}^N -\Delta_j + \sum_{1\leq i<j \leq N} v(x_i-x_j), \\
&\Lambda = \Big[-\frac{\ell}{2}, \frac{\ell}{2}\Big]^d, \qquad \ell = \frac{K_{\ell}}{\sqrt{\rho \widehat{g}(0)}}.
\end{align*}
We can now state the main theorem of the paper. 

\begin{theorem}\label{thm:main} There exists $C>0$, such that, if $v\in L^2(\Lambda)$ is a positive, spherically symmetric, compactly supported potential with scattering length $a >0$ and if $\rho > 0$ is such that $\rho a^d \leq C^{-1}$, then for any bosonic, normalized state $\Psi$ in the domain of $\mathcal{H}_{N}$ we have
	\begin{equation*}
		\langle \Psi,  \mathcal H_N  \Psi \rangle \geq \frac{1}{2}\rho^{2}\vert \Lambda \vert\widehat g(0)+ E^{\rm{LHY}}_{d}(\rho, \Lambda)+o^{\mathrm{LHY}}_{d}.
	\end{equation*}
I.e. inserting the values of $\widehat g(0)$, $E_d^{\rm{LHY}}$ and $\Ibog$,
\begin{equation}
\inf \mathrm{Spec}(\mathcal{H}_N) \geq \begin{dcases}
4 \pi \rho^2 \lvert \Lambda\rvert a \Big( 1+ \frac{128}{15\sqrt{\pi}}\sqrt{\rho a^3}\Big) + o\big( \rho^2 \lvert \Lambda \rvert a \big), \quad &\text{if } d=3,\\
4 \pi \rho^2 \lvert \Lambda\rvert \delta \Big( 1+ \Big[ 2 \Gamma + \frac{1}{2} + \log(\pi)\Big]\delta \Big) + o\big( \rho^2 \lvert \Lambda \rvert \delta \big), \quad &\text{if } d=2.
\end{dcases}
\end{equation}
\end{theorem}

\begin{remark}[Bogoliubov integral]
	The integral \eqref{eq:ibog} can be explicitly calculated and provides the expected coefficients for the LHY corrections
	\begin{equation}\label{def:Idbog}
		I^{\rm{Bog}}_{d}=\begin{dcases}
			2\Gamma +\frac{1}{2}+\log \pi, \quad &  \text{if } d=2, \\
			\frac{128}{15\sqrt{\pi}}, \quad & \text{if } d=3.
		\end{dcases}
	\end{equation}
	Notice furthermore, that the whole second order term $E_d^{\text{LHY}}$ of the energy comes from the calculation of the integral 
	\begin{equation}
		\frac{|\Lambda|}{2(2\pi)^d} \int_{\mathbb{R}^d} \Big( \sqrt{k^4 + 2 k^2 \rho \widehat g(k)} - k^2 - \rho \widehat g(k) + \rho^2  G_d(k) \Big) \dd k,
	\end{equation}
	from which we recover \eqref{def:LHY-2ndenergy} thanks to a change of variables $k \mapsto \sqrt{\rho \widehat{g}(0)}k$, and a passage to the limit $\rho a^{d}\to 0$.
	
\end{remark}
\begin{remark}[Assumptions on the potential]
The $L^2$ assumption on $v$ in Theorem~\ref{thm:main} is technical and not needed in the actual papers dealing with the thermodynamic limit \cite{FS2,2DLHY}, where $L^1$ suffices. In the present paper we need this assumption in the comparison between the discrete sums over the dual lattice and the corresponding continuous integrals (see \eqref{eq:BadEstimate} and the proof of Proposition~\ref{prop:BogDiag}). Actually, for this point the assumption $v \in L^{p}(\Lambda)$ for any $p>6/5$ would suffice.

	These $L^{p}$-assumptions on the potential $v$ exclude the hard core case. These assumptions are actually also not necessary. Indeed, the inequalities of the proof in the thermodynamic setting allow for a large $L^{1}$-norm. This is enough to extend the result to the hard core case approximating it through a sequence of growing $L^{1}$-potentials. See \cite[Theorem 1.6]{FS} and \cite[ Section 3.3]{2DLHY} for the $3D$ and $2D$-case respectively.
	
	The compact support assumption on the potential $v$ can also be relaxed in the thermodynamic regime. We can allow for a tail under a proper decay assumption provided that, avoiding the contribution from the tail does not affect the scattering length too much. See \cite[Theorem 1.6]{FS} and \cite[ Section 3.2]{2DLHY}.
\end{remark}
\begin{remark}
	The present article reviews, in the simpler setting of the periodic box, results stated in \cite[Theorem 1.3]{FS} and \cite[Theorem 2.3]{2DLHY} for the $3D$, $2D$-case respectively, neglecting the complications derived from the double localization for the thermodynamic limit. Nevertheless we included an original bound on the number of high momentum excitations \eqref{bound_n_+H}. Similar results in three dimension were proven in \cite{BrennSchlein} with different methods. 
\end{remark}
\begin{remark}
As already mentioned, the purpose of the present paper is mainly expository. The main ideas of \cite{FS, FS2, 2DLHY} are clearest in the periodic setting, which is the setting of this paper. To prove the analogous lower bound in the thermodynamic setting one would first need to localize in such periodic boxes, but it is not clear how to make such a localization with the right precision. Indeed, in \cite{FS, FS2, 2DLHY}, the localization is done by a sliding technique which produces a much more complicated kinetic energy in the boxes. 

In the papers \cite{BocSeir, HainzlNamTriay} the corresponding localization procedure is done by imposing Neumann bounday conditions which also introduces substantial technical difficulties compared to the periodic case. 
\end{remark}

\subsection{Strategy of the proof}

\begin{enumerate}
	\item \textbf{Splitting of the potential and renormalization.}
	We expect the ground state of our operator to exhibit condensation, meaning that most particles should have zero momentum. This is why we start by decomposing the potential energy according to creation or annihilation of bosons with zero and non-zero momenta. We define the following operators on $L^2(\Lambda)$, denoting by $|1\rangle$ the function which has constant value $1$ on $\Lambda$,
	\[ P = \vert \Lambda \vert^{-1} | 1 \rangle \langle 1| , \qquad Q = \one - P = \one_{(0,\infty)}( \sqrt{- \Delta}), \]
	projecting on the condensate and on excitations respectively. We recall that here $-\Delta $ is the periodic Laplacian on $\Lambda$. With this notation, the number of particles in the condensate $n_0$, and the number of excited particles $n_+$ are given by
	\[ n_0 := \sum_{j=1}^N P_j, \qquad n_+ := \sum_{j=1}^N Q_j = N - n_0, \]
	where $P_j$ and $Q_j$ denotes $P$ and $Q$ acting on the $j$-th variable. We insert these projections in the potential energy,
	\begin{equation}
		\sum_{i<j}v(x_i-x_j) = \sum_{k=0}^4 \mathcal Q_k,
	\end{equation}
	where $\mathcal Q_k$ contains precisely $k$ occurences of $Q$'s. For instance, 
	\begin{equation}
		\mathcal Q_0 = \sum_{i<j} P_i P_j v(x_i-x_j) P_j P_i.
	\end{equation}
	One should also note that $\mathcal Q_1 = 0$ by momentum conservation.

We need terms to depend on $g$ instead of $v$ in order for the scattering length to appear. We are able to overcome this problem modifying each $\mathcal{Q}_j$ into $\mathcal{Q}^{\mathrm{ren}}_j$ and collecting in the last term $\mathcal{Q}_4^{\text{ren}}$, which is positive, all the error terms produced by the renormalization. 
For a lower bound $\mathcal{Q}_4^{\text{ren}}$ can be discarded.
	\item \textbf{c-number substitution.}
	From this point on, we work in momentum space and second quantization; the operator can be rewritten in terms of creation and annihilation operators of plane waves $a_k^\dagger$, $a_k$ (Proposition \ref{prop:Hsecondquant}). The next step is a rigorous justification of the so-called \emph{c-number substitution}, which is given by expanding the operator on projectors on coherent states living in the $0$-momentum space. This allows us to replace $a_0$ and $a_0^\dagger$ by their actions as multiplication by complex numbers $z$ on the coherent states (Proposition \ref{prop:projonz}). This amounts to consider the condensate of $0$-momenta particles having fixed density $\rho_z= \vert z \vert^2 \vert \Lambda \vert^{-1}$ and to only work on the remaining degrees of freedom in the space of excitations.
	\item \textbf{Bogoliubov diagonalization.}
	We first focus on $\mathcal{Q}_{0}^{\mathrm{ren}} $ and the quadratic excitation operator $\mathcal Q_2^{\rm{ren}}$. The sum of these with the kinetic energy produces a $\mathcal K(z)$ that can be diagonalized, as in the standard Bogoliubov theory. This procedure gives rise to the Bogoliubov integral $I_d^{\rm{Bog}}$, times the LHY order, which is the second order term of the energy, together with a positive diagonal operator $\mathcal K^{\rm{diag}}$ (Proposition \ref{prop:BogDiag}). The remaining quadratic terms have to be bounded by the contribution given by the soft-pairs in $\mathcal{Q}_3^{\text{ren}}$, introduced in the next step.
	\item \textbf{Localization of 3Q terms.}
	One of the major difficulties is to deal with the 3Q terms $\mathcal Q^{\rm{ren}}_3$. These terms can be interpreted as the energy generated by one pair of excited momenta, interacting to give one zero and one excited momentum or the other way around. The upper bound calculations of \cite{YY} show that such pairs are crucial to find the correct energy to LHY precision, and especially the \emph{soft pairs}. Those pairs have high momentum, and interact to create one zero momentum and one low momentum. In fact, we show in Proposition \ref{prop:Q3loc} that $\mathcal Q_3^{\rm{ren}}$ gives almost the same contribution to the energy as the analogue soft pairs operator $\mathcal Q_3^{\rm{soft}}$.
	\item \textbf{The energy of soft pairs.}
	Section \ref{Sec:softpair} is dedicated to the bounds on $\mathcal Q_3^{\rm{soft}}$. It absorbs the remaining part of the quadratic energy $\mathcal Q_2^{\rm{ex}}$, using the high momenta part of $\mathcal K^{\rm{diag}}$. The precise understanding of the $\mathcal Q_3^{\rm{soft}}$ is a key calculation in our approach.
	\item \textbf{Bounds on the number of excitations.}
	Most of our bounds require estimates on the number of excited particles $n_+$, the number of high-momenta excited particles $n_+^H$ and the number of low momenta excited particles $n_+^L$. In \ref{sec:largeM}, we use the technique called \emph{localization of large matrices} to show that we can restrict to states having bounded $n_+^L$. In \ref{app:C}, we directly get bounds on $n_+$ and $n_+^H$, i.e., \emph{condensation estimates} on $\Lambda$.
	\item \textbf{Conclusion.}
	In the final Section~\ref{sec:concl} we combine all the estimates to finish the proof of Theorem~\ref{thm:main}.

\end{enumerate}

The proof depends on several parameters that have to be suitably tuned. These parameters and their relations are collected in \ref{app:parameters}.
\section{Splitting of the Potential Energy and Renormalization}
By means of the projectors onto and outside the condensate, we split the potential in a sum of operators by expanding
\begin{equation}
	v(x_{i}-x_{j})=(P_{i}+Q_{i})(P_{j}+Q_{j})v(x_{i}-x_{j})(P_{j}+Q_{j})(P_{i}+Q_{i})\nonumber
\end{equation}
and reorganize it as a sum of $\mathcal{Q}_j$, where in each $\mathcal{Q}_j$, the projector $Q$ is present $j$ times. An idea similar to this already appeared in the early work \cite{GrafSol}. We then renormalize   the $\mathcal{Q}_j$ to obtain $\mathcal{Q}_j^{\rm{ren}}$ where $v$ has been replaced by $g$. More precisely we have
\begin{lemma}\label{lem:potential_splitting}
The following algebraic identity holds
\begin{equation}
\frac{1}{2}\sum_{i\neq j} v(x_i - x_j) = \sum_{j=0}^4 {\mathcal Q}_j^{\rm ren},
\end{equation}
where
\begin{align}
0 \leq {\mathcal Q}_4^{\rm ren}&:=
\frac{1}{2} \sum_{i\neq j} \Big[ Q_i Q_j + \left(P_i P_j + P_i Q_j + Q_i P _j\right)\omega(x_i-x_j) \Big] v(x_i-x_j) \nonumber \\
&\,\qquad \qquad \times
\Big[ Q_j Q_i + \omega(x_i-x_j) \left(P_j P_i + P_j Q_i + Q_j P_i\right)\Big],\label{eq:SF_DefQ4}\\
{\mathcal Q}_3^{\rm ren}&:=
\sum_{i\neq j} P_i Q_j g(x_i-x_j) Q_j Q_i + h.c.,
 \label{eq:SF_DefQ3}\\
{\mathcal Q}_2^{\rm ren}&:=
\sum_{i\neq j} P_i Q_j (g+ g\omega)(x_i-x_j)P_j Q_i   + \sum_{i\neq j} P_i Q_j (g+g\omega)(x_i-x_j)Q_j P_i   \nonumber\\
&\quad+\frac{1}{2}\sum_{i\neq j} P_iP_j g(x_i-x_j) Q_j Q_i + h.c.,
 \label{eq:SF_DefQ2}\\
{\mathcal Q}_1^{\rm ren}&:=\sum_{i,j} \big(Q_i P_j (g + g \omega)(x_i-x_j) P_j P_i  + h.c.\big ) = 0,  \label{eq:SF_DefQ1}
\intertext{and}
{\mathcal Q}_0^{\rm ren}&:= \frac{1}{2} \sum_{i \neq j} P_i P_j (g+ g\omega)(x_i-x_j) P_j P_i. \label{eq:SF_DefQ0}
\end{align}
\end{lemma} 

\begin{proof}
The lemma is proven by algebraic computations using that $g = v (1 - \omega)$, and $\mathcal{Q}_1^{\rm ren}$ is zero because, for any $f \in L^1(\Lambda)$,
\begin{equation*}
Q_i P_j f(x_i-x_j) P_j P_i =\frac{1}{\lvert \Lambda \rvert} \|f\|_{L^1} Q_i P_i = 0. 
\end{equation*}
\end{proof}

We continue our analysis in momentum space considering the second quantization of the Hamiltonian. Let us introduce 
\begin{equation}
a_k^{\dagger}:= \frac{1}{|\Lambda|^{1/2}}a^{\dagger}(e^{ikx} ), \qquad a_k:= \frac{1}{|\Lambda|^{1/2}}a(e^{ikx} ),
\end{equation}
$\mathrm{i.e.}$ the usual bosonic creation and annihilation operators of bosons with momentum $k \in \Lambda^* = \frac{2 \pi }{\ell}\mathbb{Z}^d$. Note that for zero momentum, $a_0^\dagger$ creates the function $1$, the \emph{condensate} in $\Lambda$. 
The operator $\mathcal{H}_N$ can be written, by abuse of notation, as the action on the $N-$boson space of a second quantized Hamiltonian acting on the Fock space $\mathscr{F}_s(L^2(\Lambda)) = \bigoplus_{N=0}^{\infty} L_s^2(\Lambda^N)$ involving $a_k$ and $a^{\dagger}_k$. We can write the number operators as
\begin{equation}
n_0 = a^{\dagger}_0 a_0, \qquad n_+= \sum_{k \in \Lambda^*} a^{\dagger}_k a_k.
\end{equation}
\begin{proposition}\label{prop:Hsecondquant}
The Hamiltonian $\mathcal H_N$ acts on $L_{s}^2(\Lambda^N)$ as
\begin{align}\label{eq:second_quant_Ham}
\mathcal H_N &= \sum_{k \in \Lambda^*}k^2 a_k^\dagger a_k +\frac{1}{2 \vert \Lambda \vert} \big( \widehat g (0) + \widehat{g \omega}(0) \big) a_0^\dagger a_0^\dagger a_0 a_0  \nonumber \\
& \quad + \frac{1}{|\Lambda|} \sum_{k \in \Lambda^*}\Big( (\widehat g(k) + \widehat{g \omega}(k)) a_0^\dagger a_k^\dagger a_k a_0 + \frac 1 2 \widehat g(k) (a_0^\dagger a_0^\dagger a_k a_{-k} + h.c.) \Big)  \nonumber \\
& \quad + \Big( \widehat g(0)   +\widehat{g \omega}(0)    \Big) \frac{n_0 n_+}{\vert \Lambda \vert } + \mathcal Q_3^{\rm{ren}} + \mathcal Q_4^{\rm{ren}}.
\end{align}
\end{proposition}
\begin{proof}
The first term of \eqref{eq:second_quant_Ham} is obtained by a simple application of the second quantization to the Laplacian. 
The other terms require some manipulations with the $\mathcal Q^{\rm ren}_j$. We observe that 
\begin{equation}\label{eq:PgPformula}
\sum_{j=1}^n P_j g(x_i- x_j) P_j = \frac{1}{\vert \Lambda \vert} \sum_{j=1}^n P_j \int_{\Lambda} g(x_i-y) dy = \frac{n_0}{|\Lambda|} \widehat g(0).
\end{equation}
In particular $\mathcal{Q}_0^{\rm ren}$ is
\begin{equation}
\mathcal{Q}_0^{\rm ren} = \frac{n_0(n_0 -1)}{2|\Lambda|} (\widehat{g}(0) + \widehat{g\omega}(0)), 
\end{equation}
and by the second quantization we get the second term in \eqref{eq:second_quant_Ham}.
For $\mathcal{Q}_2^{\rm ren}$, we use \eqref{eq:PgPformula} for 
\begin{equation}\label{eq:PQPQ_n+estimate}
\sum_{i\neq j} P_i Q_j (g+g\omega)(x_i-x_j)Q_j P_i  =  \Big( \widehat g(0)   +\widehat{g \omega}(0)    \Big) \frac{n_0 n_+}{\vert \Lambda \vert }.
\end{equation}
The second quantization of the whole $\mathcal{Q}_2^{\rm ren}$ is obtained by a standard calculation which provides the third and fourth terms of \eqref{eq:second_quant_Ham}. We only provide here an example of this calculation for the term
\begin{equation}
\mathcal{Q}_2^{1}:=\sum_{i\neq j}P_{i}Q_{j}g(x_{i}-x_{j})P_{j}Q_{j}.
\end{equation}
We denote the basis elements $e_{p}(x)=\frac{e^{ipx}}{\sqrt{\vert\Lambda\vert}}$ and write a $\Psi\in L^{2}(\Lambda^{N})$ as
$$\Psi =\sum_{p,k}c_{pk}e_{p}(x_{j})e_{k}(x_{i})\quad \text{with}\quad c_{pk}=\frac{1}{\sqrt{N(N-1)}}a_{p}a_{k}\Psi .$$ 
We can then compute
\begin{align}
	\mathcal{Q}_2^{1}\Psi&=\frac{1}{\vert \Lambda\vert}\sum_{k\neq 0}\widehat{g}(k)\sum_{i\neq j}e_{k}(x_{j})e_{0}(x_{i})a_{0}a_{k}\Psi\\
	&=\frac{1}{\vert \Lambda\vert}\sum_{k\neq 0}\widehat{g}(k)a_{k}^{\dagger}a_{0}^{\dagger}a_{0}a_{k}\Psi.
\end{align}
\newline
\end{proof}

\section{c-Number Substitution}\label{Sec:cn}

Now that the operator is written in second quantization, as stated in Proposition~\ref{prop:Hsecondquant}, we proceed to the $c$-number substitution. Thanks to this procedure, we can turn the action of the $a_0$'s into multiplication by complex numbers $z$. It amounts to consider the condensate of $0$-momentum particles as having a fixed density $\rho_z = \vert z \vert^2 \vert \Lambda \vert^{-1}$, and only deal with excitations. This is done by diagonalizing $a_0$ in the following way. The decomposition $L^{2}(\Lambda)=\mathrm{Ran}P \oplus \mathrm{Ran}Q$ leads to the splitting of the bosonic Fock space $\mathscr{F}_s(L^{2}(\Lambda))=\mathscr{F}_s(\mathrm{Ran}P)\otimes\mathscr{F}_s(\mathrm{Ran}Q)$.
	Denoting by $\Omega$ the vacuum vector, we introduce the class of coherent states in $\mathscr F_s(\mathrm{Ran}P)$, labeled by $z \in \mathbb{C}$,
		\begin{equation}\label{eq:coherent_z}
		|z\rangle = e^{-\big(\frac{|z|^2}{2} + z a^{\dagger}_0\big)}\, \Omega,
	\end{equation}
	which are eigenvectors for the annihilation operator of the condensate. It is simple to show that
	\begin{equation}\label{eq:coherent_prop}
		a_0 |z\rangle = z\, |z \rangle \;\;\;\;\text{and}\;\;\;\;	1 = \frac{1}{\pi} \int_{\mathbb{C}} |z\rangle \langle z| \, \dd z.
	\end{equation}

Here $\langle z |$ is the partial trace along $\mathscr{F}_s(\mathrm{Ran}P)$. Thus, for any $\Psi \in \mathscr{F}_s(L^2(\Lambda))$ the state $\Phi(z) = \langle z | \Psi \rangle$ is in $\mathscr{F}_s(\mathrm{Ran} Q)$.

\begin{proposition}\label{prop:projonz}
	For $z \in \mathbb C$, set $\rho_z = \vert z \vert^2\vert \Lambda \vert^{-1}$. The Hamiltonian $\mathcal H_N$ acts on $L_{\rm{sym}}^2(\Lambda^N)$ as
	\begin{equation}\label{eq:c-numberHamiltonian}
		\mathcal H = \frac{1}{\pi}\int_{\mathbb{C}}   \mathcal K(z) |z\rangle \langle z| \dd z +\mathcal Q_3^{\rm{ren}} + \mathcal Q_4^{\rm{ren}}  + \mathcal R_0,
\end{equation}
	where the $z-$dependent Hamiltonian is 
	\begin{align}
		\mathcal K(z) &:= \mathcal Q(z) + \mathcal Q_2^{\rm{ex}}(z) + (\rho_z -\rho)n_+ \widehat{g}(0) - \rho \rho_z |\Lambda| \widehat{g}(0)+ \rho^2|\Lambda|\widehat{g}(0), \label{eqKZ0}
\end{align}
with 
\begin{equation}\label{def:Q}
\mathcal Q(z):= \frac{1}{2}    \rho_{z}^2 \vert\Lambda\vert (\widehat g(0) +    \widehat{ g\omega}(0)) + \mathcal{K}^{\rm{Bog}}, 
\end{equation}
where $\mathcal{K}^{\rm{Bog}}$ is a quadratic Hamiltonian in creation and annihilation operators that we call the \textit{Bogoliubov Hamiltonian}:
\begin{equation}
\mathcal{K}^{\rm{Bog}} =  \frac{1}{2} \sum_{k \neq 0} \mathcal{A}_k  \big( a_k^\dagger a_k + a_{-k}^\dagger a_{-k} \big)
  + \frac 
1 2 \sum_{k \neq 0} \mathcal{B}_k \big( a_k^\dagger a_{-k}^\dagger + a_k a_{-k}\big),
\end{equation}
with 
\begin{equation}
\mathcal{A}_k :=  k^2 + \rho_z \widehat{g}(k) , \qquad \mathcal{B}_k := \rho_z \widehat g(k).
\end{equation}
The remaining 2Q term is
\begin{align}
		\mathcal Q_2^{\rm{ex}}(z) &=  \rho_z  \sum_{k \neq 0} \big( \widehat{g\omega}(k) + \widehat{g \omega}(0) \big) a_k^\dagger a_k. \label{def:Q2ex}
	\end{align}
	Moreover, there exists a universal constant $C>0$ such that the error term satisfies
\begin{equation}\label{eq:R0_estimate}
\vert \langle \mathcal R_0 \rangle_\Psi \vert \leq C N\vert \Lambda\vert^{-1}\widehat{g}(0), \qquad \forall \Psi \in L_{\rm{sym}}^2(\Lambda^N) \quad \text{normalized.}
\end{equation}
\end{proposition}
\begin{proof}
As a first step, we add and subtract in the Hamiltonian the term $\rho^2 |\Lambda| \widehat{g}(0)$, exploiting the identity on $ L_{\rm{sym}}^2(\Lambda^N)$
\begin{equation}
\rho^2 |\Lambda| \widehat{g}(0) = \rho (n_0 + n_+) \widehat{g}(0).
\end{equation}
We focus then on $\mathcal{H} - \rho n_0 \widehat{g}(0)$, and apply to this term the c-number substitution, briefly described below.
The expansion on coherent states allows to perform, for instance, the following formal substitutions in \eqref{eq:second_quant_Ham}
	\begin{equation}\label{eq:c-substitutions_a0}
		a_0^{\dagger}a_0^{\dagger}a_0a_0 \mapsto |z|^{4}-4|z|^{2}+2,  \qquad a^{\dagger}_0 a_0 \mapsto |z|^2-1.
	\end{equation}	
	We give an example of the rigorous derivation of the second term in \eqref{eq:c-substitutions_a0} as follows. For any $f,g\in \mathscr{F}_s(L^{2}(\Lambda))$, using \eqref{eq:coherent_prop},
	\begin{equation}
		\langle f|a_{0}^{\dagger}a_{0}g\rangle =\langle f|a_{0}a_{0}^{\dagger}g\rangle -\langle f|g\rangle =\frac{1}{\pi}\int_{\mathbb{C}}|z|^{2}\langle f|z\rangle \langle z|g\rangle \dd z -\frac{1}{\pi}\int_{\mathbb{C}}\langle f| z \rangle \langle z|g\rangle \dd z,
	\end{equation}
yielding
\begin{equation}
a^{\dagger}_0 a_0 = \frac{1}{\pi} \int_{\mathbb{C}} (|z|^2-1) \vert z \rangle \langle z\vert \dd z,
\end{equation}
and the other terms can be treated in a similar manner. We now prove how low order terms produced in the aforementioned substitution are actually errors. For instance, focusing again on the $|z|^2$ in the first term of \eqref{eq:c-substitutions_a0}, we have that
\begin{equation}
\frac{\widehat{g}(0)}{2\pi|\Lambda|}\int_{\mathbb{C}}|z|^2 \vert z \rangle \langle z \vert \dd z = \frac{ \widehat{g}(0)}{2|\Lambda|} a_0 a^{\dagger}_0 \geq - C \frac{n_0+1}{|\Lambda|} \widehat{g}(0) .
\end{equation}
The substitution step leads to the result, with 

	\begin{align}
		\mathcal R_0&=-\frac{1}{2 \vert \Lambda \vert} \big( \widehat g (0) + \widehat{g \omega}(0) \big)\big(4n_{0}-2\big)-\widehat{g\omega} (0)\frac{n_{+}}{\vert \Lambda\vert}\\
		& \quad - \frac{1}{\vert \Lambda \vert} \sum_{k \in \Lambda^*} \Big( (\widehat g(k) + \widehat{g \omega}(k))  a_k^\dagger a_k + \frac 1 2 \widehat g_k ( a_k a_{-k} + h.c.) \Big),
	\end{align}
	and bound the error term using a Cauchy-Schwarz on the $a_{k}a_{-k}$ terms to get
	\begin{equation*}
		\vert \mathcal R_0\vert \leq C\frac{n_{0}+n_{+}}{\vert \Lambda \vert} \widehat g(0)\leq  C \frac{N}{\vert \Lambda\vert} \widehat{g}(0).
	\end{equation*}

Notice that the substitutions of $a_{0}a_{0}$ and $a^{\dagger}_{0}a^{\dagger}_{0}$ should give a $z^{2}$ and a $\overline{z}^{2}$ in the definition of  $\mathcal{B}_k := \vert z\vert^{2} \vert\Lambda\vert^{-1}\widehat g(k)$. To circumvent this issue we write $z=\vert z\vert e^{i\phi}$ and absorb the phase in the $a_{k}$'s. This does not affect the later computations which only involve commutations of such $a_{k}$'s.
	\end{proof}

By Proposition~\ref{prop:projonz} we are reduced to study a Hamiltonian dependent on the free parameter $z\in \mathbb{C}$. The density $\rho_z$ describes the particles in the condensate, but we have no restriction on it. 
We expect to have full condensation, i.e. $\rho_z \simeq \rho$. In this regime we need to make very precise estimates which are established in the main part of the paper. The regime where $\rho_z$ is far from $\rho$ seems less physical and, in fact, there rougher bounds suffice.

We define the threshold magnitude for the densities 
\begin{equation}\label{eq:defe+}
\varepsilon_+ := \max \{K_{\ell}^2 K_L^{-1}, (\f)^{1/2}\},
\end{equation}

with $K_L$ being introduced in \eqref{def:PLPH} below (and fixed in \ref{app:parameters}).
In the following sections, we will study the regime 
\begin{equation}
|\rho_z -\rho | < \rho \varepsilon_+, \qquad 
\end{equation}
while we deal with the regime $|\rho_z -\rho | \geq \rho \varepsilon_+$ in \ref{app:rho_far}.

\section{Estimates for $\rho_z$ close to $\rho$}\label{sec:close}
\subsection{Diagonalization}
We apply the diagonalization procedure to the operator
\begin{equation}
	\mathcal Q(z)=\frac{\rho_z^2}{2} \vert \Lambda \vert (\widehat{g}(0) +\widehat{g\omega}(0))+\mathcal{K}^{\rm{Bog}}
\end{equation}
defined in \eqref{def:Q} and containing the LHY integral and a positive operator, diagonal in creation and annihilation of excitations.
\begin{proposition}\label{prop:BogDiag}
Let $\varepsilon_+$ be as in \eqref{eq:defe+} and
	assume the relations between the parameters in \ref{app:parameters}. For any $z \in \mathbb C$ such that $\vert\rho -\rho_{z} \vert\leq \rho\varepsilon_{+}$, the following equality holds:
	\begin{align*}
		\mathcal Q(z) = \frac{\rho_z^2}{2}\vert \Lambda \vert  \widehat g(0)  +  E^{\rm{LHY}}_d(\rho_z) +\mathcal{K}^{\rm{diag}}+  \mathcal R^{(d)}_1, 
	\end{align*}
	where we define the \emph{diagonalized Bogoliubov Hamiltonian} as
	\begin{equation}\label{eq:Dk}
		\mathcal K^{\rm{diag}} = \sum_{k \neq 0} \mathcal D_k b_k^\dagger b_k, \qquad \mathcal D_k = \sqrt{k^4 + 2k^2 \rho_z \widehat g(k)},
	\end{equation}
	where
	\begin{equation}
		b_k =  \frac{1}{\sqrt{1-\alpha_k^2}} \Big( a_k + \alpha_k a_{-k}^\dagger \Big),  \qquad
		\label{def:bk}
		\alpha_k = \frac{k^2 + \rho_z \widehat g(k) - \sqrt{k^4 + 2k^2 \rho_z \widehat g(k)}}{\rho_z \widehat g(k)},
	\end{equation}
	and where the error term $\mathcal{R}_1^{(d)}(\rho_z)$ satisfies
	\begin{equation}
	\vert	\mathcal R_1^{(d)}(\rho_z) \vert\leq \begin{cases}
			C|\Lambda|\rho_z^{2} \delta^{2}K_{\ell}^{-1}, \quad  & \text{if } d=2, \\
			C|\Lambda| \rho_z^{2} a \big(\rho_z a^{3}\big)^{\frac 1 2}\log(\rho_z) K_{\ell}^{-1}, \quad  &\text{if } d=3.
		\end{cases}
		\label{eq:R1bound}
	\end{equation}
	
\end{proposition}
The constant in \eqref{eq:R1bound} depends on $L^p$-norms of the potential.
\begin{proof}
	Applying Theorem~\ref{thm:bogdiag} with $\mathcal A_k = k^2 + \rho_z \widehat g(k)$ and $\mathcal B_k = \rho_z \widehat g(k)$, we get $\mathcal D_k$ and $\alpha_k$ from \eqref{eq:Dk} and \eqref{def:bk},

and we can write, for all $k \neq 0$,
	\begin{align}
		\mathcal A_k (a_k^\dagger a_k &+ a_{-k}^\dagger a_{-k}) + \mathcal B_k (a_k^\dagger a_{-k}^\dagger + a_k a_{-k}) = \nonumber \\
		& \mathcal D_k (b_k^\dagger b_k + b_{-k}^\dagger b_{-k}) +  \sqrt{k^4 + 2 k^2 \rho_z \widehat g(k)} - k^2 - \rho_z \widehat g(k).
	\end{align}
	Then, using that $\mathcal A_k$ and $\mathcal B_k$ are even functions of $k$, we deduce
	\begin{align}
		\mathcal Q(z) &= \frac 1 2 \rho_z^2\vert \Lambda \vert  \widehat g(0) + \frac{1}{2}\sum_{k \neq 0} \Big( \sqrt{k^4 + 2 k^2 \rho_z \widehat g(k)} - k^2 - \rho_z \widehat g(k) \Big) \nonumber \\ &\quad+ \frac 1 2 \rho_z^2\vert \Lambda \vert  \widehat{g\omega}(0) + \mathcal K^{\rm{diag}}. \label{eq:Kzproof}
	\end{align}
	Changing the sum in \eqref{eq:Kzproof} into the integral
	\begin{equation}
	 \frac{\lvert \Lambda \rvert }{2(2\pi)^d}\int \Big( \sqrt{k^4 + 2 k^2 \rho_z \widehat g(k)} - k^2 - \rho_z \widehat g(k) \Big) \dd k,
\end{equation}
can be done up to an error term $\mathcal{R}_1^{(d)}(\rho_z)$ which can be estimated as in \eqref{eq:R1bound}. 
The constant in \eqref{eq:R1bound} depends on $L^p$-properties of the potential, since we need some decay of $\widehat g(k)$ to control the decay of the summand. This is easily achieved through an expansion of the square root and a H\"{o}lder inequality on the sum.

We recall here that $\widehat{ g \omega} (0)$ defined in \eqref{defG} can be written as an integral,
	\begin{equation}
		\frac{\rho_z^2}{2} \vert \Lambda \vert\widehat {g \omega} (0) = \rho_z^2 \vert \Lambda \vert \int \frac{\widehat g^{2}_{\R^{d}}(k) - \widehat g^{2}_{\R^{d}}(0) \one_{d}(\ell_{\delta} k)}{4 k^2} \frac{ \dd k}{(2\pi)^d}.
	\end{equation}
	The proposition follows then using Lemma~\ref{lem:calculation_LHYterm_int} to calculate the value of the integral.
\end{proof}

\section{Localization of 3Q terms}

In this section we focus on the effect of the $3Q$-term, namely
\begin{equation}
\mathcal Q_3^{\rm{ren}} = \sum_{i \neq j} P_i Q_j g(x_i-x_j) Q_i Q_j + h.c.
\end{equation}
Since $3Q$'s appear in this term, we can interpret it as the energy produced when 2 non-zero incoming momenta create 1 non-zero momentum and 1 zero momentum (or vice versa). We prove below that we can restrict this interaction to soft pairs, i.e., when two ``high'' momenta and one ``low'' momentum are involved in this process. More precisely, let us define the sets of low and high momenta by
\begin{equation}\label{def:PLPH}
\mathcal P_L = \lbrace p \in \Lambda^* , \quad 0< \vert p \vert \leq K_L \ell^{-1} \rbrace, \qquad \mathcal P_H = \lbrace k \in \Lambda^*, \quad \vert k \vert \geq K_H \ell^{-1} \rbrace,
\end{equation}
where the parameters $K_L, K_H$ are fixed in \ref{app:parameters}. The condition $K_L \ll K_H$, which is part of \eqref{KL_relations}, will ensure that these sets are disjoint. We define the localized projectors by
\begin{align}
Q_L &:= \one_{\mathcal P_L} (\sqrt{-\Delta}), & \overline{Q}_L &:= Q - Q_L = \one_{(K_L \ell^{-1}, \infty)} (\sqrt{-\Delta}),\\
\label{def:QH}
	Q_H &:= \one_{\mathcal P_H} (\sqrt{-\Delta}), & \overline{Q}_H &:= Q - Q_H = \one_{(0,K_H \ell^{-1})} (\sqrt{-\Delta}).
\end{align}
The number of high excitations, namely the number of bosons outside the condensate and with momenta not in $\mathcal{P}_L$, is
\begin{equation}
n_+^H := \sum_{j=1}^n \overline Q_{L,j},
\end{equation}
acting on $L^2_{\rm{sym}}(\Lambda^n)$ for any $n$. Similarly, we define the number of low excitations by
\begin{equation}
n_+^L := \sum_{j=1}^n \overline{Q}_{H,j} .
\end{equation}
Notice that $n_+^L+n_+^H \geq n_+$, due to the overlap of the regions in momentum space.

The reduction to soft pairs is then given by the following proposition.

\begin{proposition}\label{prop:Q3loc} Assuming the relations between the parameters in \ref{app:parameters}, there exists a universal constant $C>0$ such that, for all $N$-particle states $\Psi \in L_{\rm{sym}}^2(\Lambda^N)$ satisfying $\Psi = \one_{[0,2\mathcal{M}]}(n_+^L)\Psi$ and assumption \eqref{eq:assumption_lowE_psi}, we have
\begin{align*}
\vert \langle \mathcal Q_3^{\rm{ren}} \rangle_{\Psi} - \langle \mathcal Q_3^{\rm{soft}} \rangle_\Psi \vert & \leq \frac{1}{4}\langle \mathcal Q_4^{\rm{ren}} \rangle_{\Psi} + o_d^{\text{LHY}}.
\end{align*}
where 
\begin{equation}\label{def:Q3soft}
\mathcal Q_3^{\rm{soft}} = \frac{1}{\vert \Lambda \vert} \sum_{ \substack{k \in \mathcal P_H, \\p \in \mathcal P_L}} \widehat g(k) \big(  a_0^\dagger a_p^\dagger a_{p-k} a_k + h.c. \big).
\end{equation}
\end{proposition}

The proof of Proposition~\ref{prop:Q3loc} will follow from the Lemmas~\ref{lem:Q3loc1} and~\ref{lem:Q3loc2} below.

\begin{lemma} \label{lem:Q3loc1}
There exists a universal constant $C > 0$ such that, for all $\varepsilon_1 >0$ and all $N$-particle states $\Psi \in L_{\rm{sym}}^2(\Lambda^N)$, we have
\begin{equation}
\vert \langle \mathcal Q_3^{\rm{ren}} \rangle_\Psi  - \langle \mathcal Q_3^{\rm{low}} \rangle_\Psi \vert
 \leq \frac{1}{4} \langle \mathcal Q_4^{\rm{ren}} \rangle_\Psi  + \rho \widehat{g}(0)\Big(C\varepsilon_1 \langle n_+ \rangle_\Psi + (C+ \varepsilon^{-1}_1 ) \langle n_+^H \rangle_\Psi\Big) , 
\end{equation}
where
\begin{equation}\label{eq:defQ3low}
\mathcal Q^{\rm{low}}_3 := \sum_{i \neq j} (P_i Q_{L,j} g(x_i - x_j)Q_i Q_j + h.c.).
\end{equation}
\end{lemma}

\begin{proof}
From the definitions we have
\begin{equation}
\mathcal Q_3^{\text{ren}} - \mathcal Q_3^{\text{low}} = \sum_{i \neq j} (P_i \overline{Q}_{L,j} g(x_i -x_j)Q_i Q_j + h.c.).
\end{equation}
In the right-hand side we aim to reconstruct the $\mathcal Q_4^{\rm{ren}}$ terms as
\begin{align}
 \sum_{i \neq j} (P_i \overline{Q}_{L,j} g Q_i Q_j + h.c.)  &=   \sum_{i \neq j} P_i \overline{Q}_{L,j} g \left[Q_i Q_j + \omega (P_iP_j +P_iQ_j + Q_i P_j) \right] + h.c. \nonumber\\
 &\quad -  \sum_{i \neq j} P_i \overline{Q}_{L,j} g\omega (P_iP_j +P_iQ_j + Q_i P_j) + h.c. \label{eq:Q'estim2}
\end{align}
We use Cauchy-Schwarz inequality on both terms. Using that $g \leq v$ in the support of $v$, the first line of \eqref{eq:Q'estim2} is controlled by
\begin{align*}
C \sum_{i \neq j} P_i \overline{Q}_{L,j} g (P_i \overline{Q}_{L,j})^{\dagger} + \frac{1}{4} \mathcal Q_4^{\text{ren}}
&= C  \widehat g(0) \frac{n_0 n_+^H}{\vert \Lambda \vert} + \frac{1}{4} \mathcal Q_4^{\text{ren}}.
\end{align*}
In the second line of \eqref{eq:Q'estim2}, the $P_iP_j$ term vanishes because $\overline{Q}_{L,j} P_j = 0$. The two other terms can be estimated as above. For instance, for any $ \varepsilon_1 > 0$,
\begin{align}
\sum_{i\neq j} ( P_i \overline{Q}_{L,j}g \omega P_i Q_j  + h.c.) &\leq
 \varepsilon^{-1}_1\sum_{i\neq j} P_i \overline{Q}_{L,j} g \omega (P_i \overline{Q}_{L,j})^\dagger +  \varepsilon_1 \sum_{i \neq j } P_i Q_j g \omega P_i Q_j  \nonumber \\
&\leq \widehat g(0) \frac{n_0}{\vert \Lambda \vert} \Big(  \varepsilon^{-1}_1 n_+^H +\varepsilon_1  n_+\Big),
\end{align}
and conclude observing that $n_0 \leq N$ when applied to $\Psi$.
\end{proof}

\begin{lemma}\label{lem:Q3loc2}
There exists a universal constant $C>0$ such that, for all $\varepsilon_2 >0$ and all $N$-particles state $\Psi \in L_{\rm{sym}}^2(\Lambda^N)$ we have
\begin{equation}
 \vert \langle  \mathcal Q_3^{\rm{low}} \rangle_\Psi - \langle \mathcal Q_3^{\rm{soft}} \rangle_\Psi \vert \leq  C\rho \widehat{g}(0) \Big( \varepsilon_2 K_H^d \langle n_+ \rangle_\Psi +  \varepsilon_2^{-1} \frac{\langle n_+ n_+^L \rangle_\Psi}{N}\Big).
 \end{equation}
\end{lemma}

\begin{proof}
First of all, we can rewrite \eqref{eq:defQ3low} in second quantization,
\begin{equation}
\mathcal Q_3^{\rm{low}} = \frac{1}{\vert \Lambda \vert} \sum_{p \in \mathcal P_L, k \neq 0} \widehat g(k) \big(  a_0^\dagger a_p^\dagger a_{p-k} a_k + h.c. \big) .
\end{equation}
From the definition \eqref{def:Q3soft} of $\mathcal Q_3^{\rm{soft}}$ we deduce
\begin{equation}
   \mathcal Q_3^{\rm{low}} - \mathcal Q_3^{\rm{soft}}
 = \frac{1}{\vert \Lambda \vert} \sum_{\substack{k \in \mathcal P_H^c , k \neq 0 \\ p \in \mathcal P_L}}  \widehat g(k) \big(  a_0^\dagger a_p^\dagger a_{p-k} a_k + h.c. \big) .
\end{equation}
When applying to $\Psi$, we can use the Cauchy-Schwarz inequality with weight $ \varepsilon_2 >0$ and deduce
\begin{align}
\vert \langle  \mathcal Q_3^{\rm{low}} \rangle_\Psi - \langle \mathcal Q_3^{\rm{soft}} \rangle_\Psi \vert \leq C\frac{\widehat g(0)}{\vert\Lambda \vert} \sum_{\substack{k \in \mathcal P_H^c , k \neq 0 \\ p \in \mathcal P_L}} \big(  \varepsilon_2  \langle a_0^\dagger a_p^\dagger a_p a_0 \rangle_\Psi  + \varepsilon_2^{-1} \langle a_k^\dagger a_{p-k}^\dagger a_{p-k} a_k \rangle_\Psi \big). \label{eq:Q3soft1}
\end{align}
In the first term of \eqref{eq:Q3soft1} we recognize $n_+$ and a volume of $\mathcal P_H^c$. Similarly in the second term, the $p$-sum gives $n_+$ and the $k$-sum gives $n_+^L$ (and the remaining commutator is controlled by the other terms). Thus, 
\begin{equation}
\vert \langle  \mathcal Q_3^{\rm{low}} \rangle_\Psi - \langle \mathcal Q_3^{\rm{soft}} \rangle_\Psi \vert \leq C \widehat g(0) \bigg(  \varepsilon_2  K_H^d \frac{N \langle n_+ \rangle_\Psi}{\vert \Lambda \vert} +  \varepsilon_2^{-1} \frac{\langle n_+ n_+^L \rangle_\Psi}{\vert \Lambda \vert} \bigg),
\end{equation}
and this concludes the proof.
\end{proof}

We are now ready to prove Proposition~\ref{prop:Q3loc}.

\begin{proof}[Proof of Proposition~\ref{prop:Q3loc}]
Joining together Lemma~\ref{lem:Q3loc1} and Lemma~\ref{lem:Q3loc2}, we get that the error made approximating $\mathcal{Q}_3^{\text{ren}}$ by $\mathcal{Q}_3^{\text{soft}}$, testing on a state $\Psi$ as in the assumptions such that $n_+^L \leq \mathcal{M}$, is bounded by 
\begin{equation}
\frac{1}{4}\langle \mathcal{Q}_4^{\text{ren}}\rangle_{\Psi} + C \rho \widehat{g}(0) \Big( K_{\ell}^{-2} + K_H^{d/2}\Big(\frac{\mathcal{M}}{N}\Big)^{1/2}\Big) \langle n_+\rangle_{\Psi} + C \rho \widehat{g}(0) K_{\ell}^2 \langle n_+^H\rangle_{\Psi}
\end{equation}
where we chose $\varepsilon_1 =  K_{\ell}^{-2}$ and $\varepsilon_2 = \Big(\frac{\mathcal{M}}{N K_H^d}\Big)^{1/2}$.
Let us focus on the $n_+$ terms. We use \eqref{eq:condensationestimate} of Theorem~\ref{thm:apriori_n+} to bound $\langle n_+ \rangle_{\Psi}$ and \eqref{M-KH_relation} to conclude that the expression is of an order smaller than LHY. For the $n_+^H$ terms, we use \eqref{bound_n_+H} instead and \eqref{KL_relations}.

\end{proof}

\section{Bounds on $\mathcal Q_3$ when $\rho_z \simeq \rho$ : the effect of Soft Pairs}\label{Sec:softpair}

In this section we explain the effects of soft pairs on the energy in the case when $\rho_z$ is close to $\rho$. In the remaining part of this section, we only assume that $\vert \rho_z - \rho \vert \leq \frac 1 2 \rho$, so that we can replace $\rho_z$ by $\rho$ in error estimates. 

We will see how $\mathcal Q_3^{\rm{soft}}$, $\mathcal Q_2^{\rm{ex}}$ and $\mathcal K^{\rm{diag}}$ can be combined together, as stated in Proposition~\ref{prop:Q3z} below. Actually, only the high momenta in $\mathcal K^{\rm{diag}}$ are needed, namely
\begin{equation}\label{def:KHdiag}
\mathcal K_H^{\rm{diag}} = \sum_{k \in \mathcal P_H} \mathcal D_k b_k^\dagger b_k.
\end{equation}
Note that we can use $c$-number substitution to rewrite $\mathcal Q_3^{\rm{soft}}$ as
\begin{equation}
\mathcal Q_3^{\rm{soft}} = \int _{\mathbb{C}}\mathcal Q_3^{\rm{soft}}(z) \vert z \rangle \langle z \rvert \dd z,
\end{equation}
with
\begin{equation}\label{eq:Q3soft.z}
\mathcal Q_3^{\rm{soft}}(z) = \frac{1}{\vert \Lambda \vert} \sum_{k \in \mathcal P_H, p \in \mathcal P_L} \widehat g(k) \big(  \bar z a_p^\dagger a_{p-k} a_k + h.c. \big) .
\end{equation}
With this notation, we prove the following result. 

\begin{proposition}\label{prop:Q3z} There exists a universal constant $C>0$ such that the following holds. Let $\rho a^d \leq C^{-1}$ and $z \in \mathbb C$ be such that $\vert \rho_z - \rho \vert \leq \frac 1 2 \rho$. Then for any normalized state $\Phi \in \mathscr F_s(\rm{Ran} Q)$ satisfying
\[ \Phi = \one_{[0,\mathcal M]}(n_+^L) \Phi, \]
we have, for a small fraction $\varepsilon_{\text{gap}}$ of the spectral gap, suitably chosen in~\ref{app:parameters} with the other parameters, 
\begin{equation}
 \langle  \mathcal Q_3^{\rm{soft}}(z) +  \mathcal K^{\rm{diag}}_H + \mathcal Q_2^{\rm{ex}}(z) \rangle_\Phi  \geq - \varepsilon_{\rm{gap}} \frac{\langle n_+ \rangle_\Phi}{\ell^2} - K_\ell^2 \frac{\langle n_+^H \rangle_\Phi}{\ell^2} + o^{\rm{LHY}}_d.
\end{equation}
\end{proposition}

In order to prove Proposition~\ref{prop:Q3z}, we start by rewriting $\mathcal Q_3^{\rm{soft}}(z)$ in terms of the $b_k$'s defined in \eqref{def:bk}. Note that
\begin{equation}
 a_k = \frac{b_k -\alpha_k b_{-k}^\dagger}{\sqrt{1- \alpha_k^2 \vphantom{\alpha_{p-k}^2}}}, \qquad a_{p-k} =  \frac{b_{p-k} -\alpha_{p-k} b_{k-p}^\dagger}{\sqrt{1- \alpha_{p-k}^2}}.
\end{equation}
Therefore,
\[a_{p-k} a_k = \frac{\left( b_{p-k} b_k - \alpha_k b_{p-k} b_{-k}^\dagger - \alpha_{p-k} b_{k-p}^\dagger b_k + \alpha_{p-k} \alpha_k b_{k-p}^\dagger b_{-k}^\dagger \right)}{\sqrt{1-\alpha_k^2 \vphantom{\alpha_{p-k}^2}}\sqrt{1-\alpha_{p-k}^2}} ,\]
and $\mathcal Q_3^{\rm{soft}}(z) = \mathcal Q_3^{(1)} + \mathcal Q_3^{(2)} + \mathcal Q_3^{(3)} + \mathcal Q_3^{(4)}$
where
\begin{align}\label{defQ31}
\mathcal Q_3^{(1)} &= \frac{1}{\vert \Lambda \vert} \sum_{\substack{k \in \mathcal{P}_H,\\ p \in \mathcal P_L}}  \frac{ \widehat g(k)}{\sqrt{1-\alpha_k^2\vphantom{\alpha_{p-k}^2}}\sqrt{1-\alpha_{p-k}^2}} \big( \bar z a_p^\dagger b_{p-k} b_k + \alpha_k \alpha_{p-k} \bar z a_p^\dagger b_{k-p}^\dagger b_{-k}^\dagger +  h.c. \big), \\
\label{defQ32}
 \mathcal Q_3^{(2)} &= -\frac{1}{\vert \Lambda \vert} \sum_{\substack{k \in \mathcal{P}_H,\\ p \in \mathcal P_L}}   \frac{\widehat g(k) \alpha_k}{\sqrt{1-\alpha_k^2\vphantom{\alpha_{p-k}^2}} \sqrt{1-\alpha_{p-k}^2}} \big( \bar z a^\dagger_p b^\dagger_{-k} b_{p-k} + z b_{p-k}^\dagger b_{-k} a_p \big),\\
 \label{defQ33}
\mathcal Q_3^{(3)} &= -\frac{1}{\vert \Lambda \vert}\sum_{\substack{k \in \mathcal{P}_H,\\ p \in \mathcal P_L}}  \frac{ \widehat g(k) \alpha_{p-k}}{\sqrt{1-\alpha_k^2\vphantom{\alpha_{p-k}^2}} \sqrt{1-\alpha_{p-k}^2}} \big( \bar z a^\dagger_p b^\dagger_{k-p} b_{k} + z b_{k}^\dagger b_{k-p} a_p \big),\\
\label{defQ34}
\mathcal Q_3^{(4)} &= - \frac{1}{\vert \Lambda \vert}  \sum_{\substack{k \in \mathcal{P}_H,\\ p \in \mathcal P_L}} \frac{ \widehat g(k) \alpha_k}{\sqrt{1-\alpha_k^2\vphantom{\alpha_{p-k}^2}} \sqrt{1-\alpha_{p-k}^2}}  [b_{p-k} , b_{-k}^\dagger] ( \bar z a_p^\dagger + z a_p)=0.
\end{align}
Notice that $\mathcal Q_3^{(4)}$ cancels due to the commutation relation $[b_{p-k} , b_{-k}^\dagger]=\delta_{-k,p-k}$.
In Lemmas~\ref{lem.Q31} and~\ref{lem.Q32} below, we get bounds on $\mathcal Q_3^{(1)}$, $\mathcal Q_3^{(2)}$, and $\mathcal Q_3^{(3)}$, thus proving Proposition~\ref{prop:Q3z}.

\subsection{Estimates on $\mathcal Q_3^{(1)}$}

The first part $\mathcal Q_3^{(1)}$ absorbs $\mathcal Q_2^{\rm{ex}}$ using $\mathcal (1-\varepsilon_{K})K_H^{\rm{diag}}$ for some parameter $\varepsilon_{K}$ chosen in \ref{app:parameters}. The remaining fraction $\varepsilon_{K}K_H^{\rm{diag}}$ will be later in the proof to control other terms.

\begin{lemma} \label{lem.Q31}
There exists a universal constant $C>0$ such that the following holds. If $\rho a^d \leq C^{-1}$, $\vert \rho_z - \rho \vert \leq \frac 1 2 \rho$, and if the parameters $\varepsilon_K, \varepsilon_{\rm{gap}} \ll 1$ and $\mathcal M >0$, satisfy the relations in~\ref{app:parameters}, then for any normalized state $\Phi \in \mathscr F_s(\rm{Ran} Q)$ satisfying
\[ \Phi = \one_{[0,\mathcal M]}(n_+^L) \Phi, \]
we have
\begin{align*}
 \big\langle \mathcal Q_3^{(1)} + \mathcal Q_2^{\rm{ex}} + \big(1-  \varepsilon_K \big) \mathcal K^{\rm{diag}}_H \big\rangle_{\Phi}
 \geq - \varepsilon_{\rm{gap}} \frac{\langle n_+ \rangle_{\Phi}}{\ell^2} - K_\ell^2 \frac{\langle n_+^H \rangle_\Phi}{\ell^2} + o^{\rm{LHY}}_d.
\end{align*}
\end{lemma}

\begin{proof}
We first reorder the creation and annihilation operators, applying a change of variables $k \mapsto -k, p \mapsto -p$ in the $\alpha$ terms,
\begin{align*}
\mathcal Q_3^{(1)} &=\frac{1}{\vert \Lambda \vert} \sum_{k \in \mathcal{P}_H ,p \in \mathcal P_L} \frac{ \widehat g(k)}{\sqrt{1-\alpha_k^2 \vphantom{\alpha_{p-k}^2}}\sqrt{1-\alpha_{p-k}^2}} \\
& \quad \times \left( \bar z a_p^\dagger b_{p-k} b_k + \alpha_k \alpha_{p-k} \bar z a_{-p}^\dagger b_{p-k}^\dagger b_{k}^\dagger + z b_k^\dagger b_{p-k}^\dagger a _p + \alpha_k \alpha_{p-k} z b_k b_{p-k} a_{-p} \right) \\
&= \frac{1}{\vert \Lambda \vert} \sum_{k \in \mathcal{P}_H ,p \in \mathcal P_L} \frac{ \widehat g(k)}{\sqrt{1-\alpha_k^2 \vphantom{\alpha_{p-k}^2}}\sqrt{1-\alpha_{p-k}^2}} \Big( \big( \bar z a_p^\dagger b_{p-k} + \alpha_k  \alpha_{p-k} z b_{p-k} a_{-p} \big) b_k  \\ 
& \quad + b_k^\dagger \big( z b_{p-k}^\dagger a_p + \alpha_k \alpha_{p-k} \bar z a^\dagger_{-p} b_{p-k}^\dagger \big) + \alpha_k \alpha_{p-k} \big( z \left[ b_k , b_{p-k} a_{-p} \right]  + \bar z \big[ a _{-p}^\dagger b_{p-k}^\dagger, b_k^\dagger \big] \big) \Big) .
\end{align*}
Note that the two last commutators vanish. Thus, we can complete the square to get, 
\begin{equation}\label{eq.Q31.low}
\mathcal Q_3^{(1)} + (1- \varepsilon_K) \mathcal K^{\rm{diag}}_H =  (1-  \varepsilon_K) \sum_{k \in \mathcal{P}_H}  \mathcal D_k c_k^\dagger c_k + \sum_{k \in \mathcal{P}_H}\mathcal T(k), 
\end{equation}
where we keep a small portion of $\mathcal{K}^{\rm{diag}}_H$ in order to bound other error terms, and we define  
\begin{equation}
 c_k = b_k + \frac{1}{\mathcal D_k (1- \varepsilon_K) \vert \Lambda \vert} \sum_{p \in \mathcal P_L} \frac{ \widehat g(k)}{\sqrt{ 1-\alpha_{k}^{2}\vphantom{\alpha_{p-k}^2} }\sqrt{1-\alpha_{p-k}^2}} \Big( z b_{p-k}^\dagger a_p + \alpha_k \alpha_{p-k} \bar z a_{-p}^\dagger b_{p-k}^\dagger \Big) ,
\end{equation}
\begin{align}\label{term.T2}
\mathcal T (k) &= - \frac{\widehat g(k)^2}{(1-\varepsilon_K) \mathcal D_k (1-\alpha_k^2)\vert \Lambda \vert^2}  \sum_{p,s \in \mathcal P_L} \frac{1}{\sqrt{1-\alpha_{s-k}^2} \sqrt{1- \alpha_{p-k}^2}} \nonumber \\ & \qquad \times \big( \bar z a_{p}^\dagger b_{p-k}  + \alpha_k \alpha_{p-k} z b_{p-k} a_{-p} \big) \big(z b_{s-k}^\dagger a_s + \alpha_k \alpha_{s-k} \bar z a_{-s}^\dagger b_{s-k}^\dagger \big) .
\end{align}

The positive $c_k^\dagger c_k$ term in \eqref{eq.Q31.low} can be dropped for a lower bound, and we can focus on the remaining term $\mathcal T (k)$. One can write
\begin{align*}
\bar z a_p^\dagger b_{p-k} + \alpha_k \alpha_{p-k} z b_{p-k} a_{-p} &= \bar z a_p^\dagger b_{p-k} + \alpha_k \alpha_{p-k} z a_{-p} b_{p-k} + \alpha_k \alpha_{p-k} z [ b_{p-k}, a_{-p} ] ,
\end{align*} 
and the last commutator vanishes. Therefore
\begin{align*}
\mathcal T(k) = - & \frac{\widehat g(k)^2 }{(1-\varepsilon_K) \mathcal D_k (1-\alpha_k^2) \vert \Lambda \vert^2} \sum_{p,s \in \mathcal P_L} \frac{1}{\sqrt{1-\alpha_{p-k}^2}\sqrt{1-\alpha_{s-k}^2}} \\ 
& \qquad \times ( \bar z a_p^\dagger + \alpha_k \alpha_{p-k} z a_{-p}) b_{p-k} b_{s-k}^\dagger ( z a_s + \alpha_k \alpha_{s-k} \bar z a _{-s}^\dagger ).
\end{align*}
Now we use a commutator to write $\mathcal T = \mathcal T_{\rm{op}} + \mathcal T_{\rm{com}}$ in normal order for the $b_k$. Since $[ b_{p-k}, b_{s-k}^\dagger ] = \delta_{s,p}$ we get
\begin{align}
\mathcal T_{\rm{op}}(k) = & - \frac{ \widehat g(k)^2}{ (1- \varepsilon_K) \mathcal D_k (1-\alpha_k^2) \vert \Lambda \vert^2} \sum_{p,s \in \mathcal P_L} \frac{1}{\sqrt{1-\alpha_{p-k}^2}\sqrt{1-\alpha_{s-k}^2}} \nonumber \\
& \qquad \times ( \bar z a_p^\dagger + \alpha_k \alpha_{p-k} z a _{-p}) b_{s-k}^\dagger b_{p-k} ( z a _{s} + \alpha_k \alpha_{s-k} \bar z a_{-s}^\dagger), \\ 
\mathcal T_{\rm{com}}(k) = & - \frac{\widehat g(k)^2}{ (1-\varepsilon_K) \mathcal D_k (1-\alpha_k^2)\vert \Lambda \vert^2} \sum_{p \in \mathcal P_L}  \frac{\vert z \vert^2}{1-\alpha_{p-k}^2}  \nonumber \\
& \qquad \times ( a_p^\dagger + \alpha_k \alpha_{p-k} a _{-p}) ( a _{p} + \alpha_k \alpha_{p-k} a_{-p}^\dagger) . \label{term.T2com}
\end{align}
$\bullet$ In order to estimate the error term $\mathcal T_{\rm{op}}$, we introduce
\begin{equation}
\tau_s := z a _{s} + \alpha_k \alpha_{s-k} \bar z a_{-s}^\dagger.
\end{equation}
In $\mathcal T_{\rm{op}}$ we commute the $b$'s through the $a$'s, $\tau_p^\dagger b_{s-k}^\dagger b_{p-k} \tau_s = b_{s-k}^\dagger \tau_p^\dagger \tau_s b_{p-k}$, since the commutators vanish in our range of indices. We use the Cauchy-Schwarz inequality
\begin{align*}
\tau_p^\dagger b_{s-k}^\dagger b_{p-k} \tau_s \leq \frac 1 2( b_{s-k}^\dagger \tau_p^\dagger \tau_p b_{s-k} +  b_{p-k}^\dagger \tau ^\dagger_s \tau_s b_{p-k} ) .
\end{align*}
Inserting this in $\mathcal T_{\rm{op}}$,  bounding $(1- \varepsilon_K)(1-\alpha_k) \geq 1/2$ for $k \in \mathcal P_H$ (by Lemma~\ref{Lem:alphak}), and noticing that we can exchange $s$ and $p$ in the sum, we find
\begin{align*}
\vert \langle \mathcal T_{\rm{op}}(k) \rangle_\Phi \vert \leq  C \frac{ \widehat g(k)^2}{\mathcal D_k \vert \Lambda \vert^2 } \sum_{p,s \in \mathcal P_L} \frac{1}{\sqrt{1-\alpha_{p-k}^2}\sqrt{1-\alpha_{s-k}^2}} \vert \langle b_{s-k}^\dagger \tau_p^\dagger \tau_p b_{s-k} \rangle_\Phi \vert.
\end{align*}
For states $\Phi$ satisfying $\one_{[0,\mathcal M]}(n_+^L) \Phi = \Phi$ we get, bounding each $\tau_s^{\dagger}\tau_s$ by $C \lvert z \rvert^2 a^{\dagger}_s a_s$ directly or by a means of Cauchy-Schwarz inequality and a change of variables, by
\begin{align*}
\vert \langle \mathcal T_{\rm{op}}(k) \rangle_{\Phi} \vert \leq  C \frac{  \widehat g(k)^2}{\mathcal D_k \vert \Lambda \vert^2} \vert z \vert^2 \mathcal M \sum_{s \in \mathcal P_L}  \langle b_{s-k}^\dagger b_{s-k}  \rangle_\Phi  .
\end{align*}
Finally, using \eqref{eq:alpha_highcontrol},
\begin{align}
\sum_{k \in \mathcal{P}_H} \vert \langle \mathcal T_{\rm{op}}(k) \rangle_{\Phi} \vert &\leq C \rho_z \ell^{4-d} \widehat g(0)^2 K_H^{-2} K_L^d \mathcal M \frac{\langle n_+^H \rangle_{\Phi} }{\ell^2}. \label{eq.T2op}
\end{align}
This term can be absorbed in $K_\ell^2 \ell^{-2} n_+^H$, as long as the relation \eqref{eq:rel_epsK_KM2} holds.

$\bullet$ We now turn to $\mathcal T_{\rm{com}}$ given in \eqref{term.T2com}. This term will absorb $\mathcal Q_2^{\rm{ex}}$. Using the Cauchy-Schwarz inequality we have
\begin{align*}
\vert \langle ( a_p^\dagger + \alpha_k \alpha_{p-k} a _{-p}) &( a _{p} + \alpha_k \alpha_{p-k} a_{-p}^\dagger) \rangle_{\Phi} - \langle a_p^\dagger a_p \rangle_{\Phi}  \vert \\
&\leq C \vert \alpha_k \alpha_{p-k} \vert\langle a_{-p}^\dagger a_{-p} + a_p^\dagger  a_p \rangle_{\Phi}   + \vert \alpha_k \alpha_{p-k}\vert^2.
\end{align*}
We deduce that
\begin{equation}
\sum_{k \in \mathcal{P}_H} \mathcal T_{\rm{com}}(k) = - \frac{1}{\vert \Lambda \vert^2} \sum_{k \in \mathcal{P}_H, p \in \mathcal P_L} \frac{\vert z \vert^2 \widehat g(k)^2 }{ (1-\varepsilon_K) \mathcal D_k} a_p^\dagger a_p + \mathcal E, \label{eq:T2com00}
\end{equation}
where (using in particular Lemma~\ref{Lem:alphak})
\begin{align}
\vert \langle \mathcal E \rangle_\Phi \vert &\leq \frac{C}{\vert \Lambda \vert^2} \sum_{k \in \mathcal{P}_H, p \in \mathcal P_L} \frac{\vert z \vert^2 \widehat g(k)^2 }{ \mathcal D_k} \big(\vert \alpha_k \alpha_{p-k} \vert \langle a_p^\dagger a_p \rangle_\Phi + \vert \alpha_k \alpha_{p-k} \vert^2 \big) \nonumber \\ &\leq C \rho_z^3 \widehat g(0)^4 \ell^{6-d} K_H^{d-6} \langle n_+ \rangle_\Phi + \widehat g(0) \vert \Lambda \vert^{-1} K_L^d K_\ell^{10} K_H^{d-10}. \label{error.E}
\end{align}
The first term in \eqref{error.E} can be absorbed in a fraction of the spectral gap if $\rho_z^3 \widehat g(0)^4 \ell^{8-d} K_H^{d-6} \ll \varepsilon_{\rm{gap}}$ using~\ref{app:parameters}, the second term is smaller than LHY by \eqref{KL_relations}.
For the main term in \eqref{eq:T2com00} we do several approximations. First,
\begin{equation}
\sum_{k \in \mathcal{P}_H} \mathcal T_{\rm{com}}(k) = - \big(1 + \mathcal O(\varepsilon_K + \ell^2 \rho \widehat g(0) K_H^{-2}) \big) \frac{\rho_z}{\vert \Lambda \vert} \sum_{k \in \mathcal{P}_H}  \frac{ \widehat g(k)^2}{ k^2 } \sum_{p \in \mathcal P_L} a_p^\dagger a_p + \mathcal E,
\end{equation}
where we used \eqref{eq:alpha_highcontrol}. 
Second, the $k$-sum is an approximation of $2 \vert \Lambda \vert \widehat{g \omega}(0)$ by Lemma~\ref{lem:gomega.approx}, and thus
\begin{equation} \label{eq:T2com01}
\sum_{k \in \mathcal{P}_H} \mathcal T_{\rm{com}}(k) = - 2 \rho_z \widehat{g \omega}(0) \sum_{p \in \mathcal P_L} a_p^\dagger a_p + \mathcal E' + \mathcal E,
\end{equation}
with $\vert \mathcal E' \vert \leq C ( \varepsilon_K \widehat g(0) + \ell^2 \rho_z \widehat g(0) K_H^{-2} + \widehat g(0)^2 K_H^{-1} + \mathcal E_d) \rho_z  n_+ $. This error is absorbed in the spectral gap $\varepsilon_{\rm{gap}} n_+ \ell^{-2}$ using \eqref{eq:epsgap2}.
Then, for $p \in \mathcal P_L$, we can replace $\widehat{g\omega}(0)$ by $\widehat{g\omega}(p)$,
\begin{equation}
\sum_{k \in \mathcal{P}_H} \mathcal T_{\rm{com}}(k) = - \rho_z  \sum_{p \in \mathcal P_L} \big( \widehat{g \omega}(0) + \widehat{g \omega}(p) \big) a_p^\dagger a_p + \mathcal E''+ \mathcal E' + \mathcal E,
\end{equation}
with error $\vert \mathcal E '' \vert \leq C R^2 \ell^{-2} K_L^2 \rho_z \widehat g(0) n_+$, absorbed in the spectral gap again by \eqref{eq:epsgap3}. Finally, if we add $\mathcal Q_2^{\rm{ex}}$ defined in \eqref{def:Q2ex}, we get a sum on $\mathcal P_L^c$ which can be bounded by $n_+^H$,
\begin{equation}
\Big \vert \sum_{k \in \mathcal P_H } \langle \mathcal T_{\rm{com}}(k) \rangle_\Phi + \langle \mathcal Q_2^{\rm{ex}} \rangle_\Phi \Big \vert \leq C \rho_z \widehat g(0) \langle n_+^H \rangle_\Phi + \vert \langle \mathcal E + \mathcal E' + \mathcal E'' \rangle_\Phi \vert,
\end{equation}
and this concludes the proof of Lemma~\ref{lem.Q31}.

\end{proof}

\subsection{Estimates on $\mathcal Q_3^{(2)}$ and $\mathcal Q_3^{(3)}$}
Here we show the remaining $\varepsilon_{K}K_H^{\rm{diag}}$ can control $\mathcal Q_3^{(2)}$ and $\mathcal Q_3^{(3)}$. 
\begin{lemma} \label{lem.Q32}
	There exists a universal constant $C>0$ such that the following holds. If $\rho a^d \leq C^{-1}$, $\vert \rho_z - \rho \vert \leq \frac 1 2 \rho$, and if the parameters satisfy the relations in~\ref{app:parameters}, then for all normalized states $\Phi \in \mathscr F_s(\rm{Ran} Q)$ 
	satisfying
	\begin{equation}\label{eq:PhiM0}
		 \Phi = \one_{[0,\mathcal M]}(n_+^L) \Phi,
	\end{equation}
	 we have
	\begin{align*}
	\Big\vert	\Big\langle \mathcal Q_3^{(2)} + \mathcal Q_3^{(3)} \Big\rangle_\Phi	\Big\vert \leq   \varepsilon_K \langle \mathcal K^{\rm{diag}}_H \rangle_\Phi
	\end{align*}
\end{lemma}

\begin{proof}
	Notice that $\mathcal Q_3^{(2)}$ and $\mathcal Q_3^{(3)}$ are identical except for the substitution of $-k$ by $k-p$, so we can focus on $\mathcal Q_3^{(3)}$. We can commute the creation operators to write this term as
	\begin{equation}
		\mathcal Q_3^{(3)} =-\frac{1}{\vert \Lambda \vert}\sum_{k \in \mathcal{P}_H, p \in \mathcal P_L}  \frac{\widehat g(k) \alpha_{p-k}}{\sqrt{1-\alpha_k^2 \vphantom{\alpha_{p-k}^2}} \sqrt{1-\alpha_{p-k}^2}} \big( \bar z  b^\dagger_{k-p} a^\dagger_p b_{k} + z b_{k}^\dagger a_p b_{k-p}  \big),
	\end{equation}
	We use the Cauchy-Schwarz inequality with weight $\varepsilon >0$, and by \eqref{eq:alpha_highcontrol},
	\begin{align*}
		\vert \langle \mathcal Q_3^{(3)} \rangle_\Phi \vert &\leq \frac{\vert z\vert}{\vert \Lambda \vert}\sum_{k \in \mathcal{P}_H, p \in \mathcal P_L}  \frac{\vert \widehat g(k) \alpha_{p-k}\vert}{\sqrt{1-\alpha_k^2 \vphantom{\alpha_{p-k}^2}} \sqrt{1-\alpha_{p-k}^2}}  \langle \varepsilon b_{k-p}^\dagger  a_p^\dagger a_p b_{k-p} + \varepsilon^{-1} b_k^\dagger b_k \rangle_\Phi  \\
		&\leq C \frac{\vert z\vert}{\vert \Lambda \vert} \ell^2 \rho_z \widehat g(0)^2 K_H^{-2} \sum_{k \in \mathcal{P}_H, p \in \mathcal P_L}  \langle \varepsilon b_{k-p}^{\dagger} a_p^\dagger  a_p b_{k-p} + \varepsilon^{-1} b_k^\dagger b_k \rangle_\Phi,
	\end{align*}
	and using \eqref{eq:PhiM0}, 
	\begin{equation}
		 \sum_{p\in \mathcal P_L}\langle b_{k-p}^\dagger  a_p^\dagger  a_p b_{k-p}  \rangle_\Phi \leq C  \mathcal M \langle  b_{k}^\dagger b_{k} \rangle_\Phi.
	\end{equation}
	We choose $\varepsilon = \sqrt{K_{L}^{d}/ \mathcal M} $, and insert $\mathcal{D}_k \geq K_H^2 \ell^{-2}$, obtaining
	\begin{align}
		\vert \langle \mathcal Q_3^{(3)} \rangle_\Phi \vert & \leq C \vert z \vert \ell^{2-d} \rho_z \widehat g(0)^2 K_H^{-2} (\varepsilon \mathcal M + \varepsilon^{-1} K^{d}_{L} ) \sum_{k \in \mathcal P_H} \langle b_k^\dagger b_k \rangle_\Phi \nonumber \\
		&\leq  C \vert z \vert \ell^{4-d} \rho_z \widehat g(0)^2  K_H^{-4}K^{d/2}_{L} \sqrt{\mathcal M }\sum_{k \in \mathcal P_H} \mathcal{D}_k\langle b_k^\dagger b_k \rangle_\Phi.
	\end{align}
	Thanks to condition \eqref{eq:rel_epsK_KM3}, $ \mathcal Q_3^{(3)}$ can be absorbed in the positive $\varepsilon_K\mathcal{K}^{\rm{diag}}_H$ term.
\end{proof}

\section{Conclusion}\label{sec:concl}
In all this section, we assume that all our parameters satisfy the relations in \ref{app:parameters}, and prove Theorem~\ref{thm:main} by combining as follows all the previous estimates.

Let us first fix $C_B \geq 2 \Ibog$, and assume that there exists a normalized $N$-particle state $\Psi \in L^2_{\rm{sym}}(\Lambda^N)$ with energy 
\begin{equation}
	\langle \mathcal H \rangle_\Psi \leq \frac{1}{2} \rho^2 \vert \Lambda \vert\widehat g(0) (1 + C_B \f).
\end{equation}
If $\Psi$ does not exist we are clearly done.

For such a state $\Psi$
we use the localization of large matrices Lemma~\ref{lem:localization_largeMat} to decompose $\Psi$ into $\Psi^{m}$'s satisfying that,
	\begin{equation}
		\Psi^{m} = \one_{\{n_+^L \leq  \frac{\mathcal{M}}{2} +m\}}\Psi^{m},\qquad \text{ and } \qquad \sum_m \Vert \Psi^m \Vert^2 = 1
	\end{equation}
	with 
	\begin{equation}
		\langle \Psi, \mathcal{H} \Psi \rangle  \geq  \sum_{2 |m|\leq \mathcal{M} }\langle \Psi^{m}, \mathcal{H} \Psi^m\rangle +  \frac{\vert \Lambda \vert}{2} \rho^2 \widehat g(0)\Big(1 + 2C_B  \lambda_d^{\rm{LHY}} \Big)\sum_{2|m| >\mathcal{M}} \|\Psi^m\|^2 +  o^{\rm{LHY}}_d.
	\end{equation}
	The next goal is then to prove our lower bound for each term of the first sum of above to reconstruct $\sum_{m} \|\Psi^m\|^2=1$. Hence we only have to prove the desired lower bound for states $\Psi \in L^2_{\rm{sym}}(\Lambda^N)$ satisfying
\begin{equation}
	\Psi = \one_{\{n_+^L \leq  \mathcal{M}\}}\Psi.
\end{equation}

For such a $\Psi$, we use the second quantization from Proposition~\ref{prop:Hsecondquant}, the c-number substitution from Proposition~\ref{prop:projonz} and the localization of the $3Q$ term in Proposition~\ref{prop:Q3loc} to deduce
\begin{equation}\label{eq:intKz}
	\langle \mathcal H \rangle_\Psi
	\geq \frac{1}{\pi} \int_{\mathbb C} \langle \mathcal{K}(z)+\mathcal{Q}_{3}^{\mathrm{soft}}(z) \rangle_{\Phi(z)} \dd z 
	+ o_d^{\rm{LHY}},
\end{equation}
where $\Phi(z) = \langle \Psi \vert z \rangle \in \mathscr F_s(\rm{Ran}Q)$ was introduced in Section~\ref{Sec:cn}. Note that we dropped the remaining part of $\mathcal{Q}_{4}^{\mathrm{ren}}>0$, and that the error terms are estimated using  Theorem~\ref{thm:apriori_n+}. Now we split the integral according to the values of $\rho_z$. We recall that $\varepsilon_+^2 = \max \lbrace K_\ell^4 K_L^{-2} , \lambda_d^{\rm{LHY}} \rbrace$ and consider the two following cases.

\begin{itemize}
	\item If $\vert \rho_z - \rho \vert \geq \rho \varepsilon_+$, we can apply Theorem~\ref{thm:rhofar} to get a lower bound larger than the LHY term, since $E_d^{\text{LHY}} >0$, i.e.
	\begin{align}\label{lowbound_rhofar}
		\langle \mathcal{K}(z)&+\mathcal{Q}_{3}^{\mathrm{soft}}(z) \rangle_{\Phi(z)} \nonumber \\
		&\geq \Big(\frac{1}{2}\rho^2 |\Lambda| \widehat{g}(0) + 2 E_d^{\text{LHY}}+ o^{\text{LHY}}_d\Big)\|\Phi(z)\|^2 -C \rho\widehat{g}(0)\langle n_+^H\rangle_{\Phi(z)}.
	\end{align}
		
	The integral of the last term over $\{z \in \mathbb{C}: \lvert \rho_z - \rho\rvert \geq \varepsilon_+ \rho\}$ can be bounded by the integral over all of ${\mathbb C}$, giving $C \rho \widehat{g}(0)\langle n_+^H\rangle_{\Psi}$ that, thanks to \eqref{bound_n_+H}, is of order $o_d^{\text{LHY}}$.
	\item Now we want to prove the desired lower bound for $\vert \rho_z - \rho \vert \leq \rho \varepsilon_{+}$. Recall that $\mathcal K(z)$ is given by
\[\mathcal K(z) = \mathcal Q(z) + \mathcal Q_2^{\rm{ex}}(z) + (\rho_z -\rho)n_+ \widehat{g}(0) - \rho \rho_z |\Lambda| \widehat{g}(0)+ \rho^2|\Lambda|\widehat{g}(0) + o_d^{\text{LHY}},\]
where we have omitted the error term $\mathcal{R}_1^{(d)}(\rho_z)$, which is lower order when $\rho_z \approx \rho$.
We diagonalize $\mathcal Q(z)$ with Proposition~\ref{prop:BogDiag} to get
	\begin{align}
		 &\mathcal{K}(z)- o_d^{\text{LHY}}\\& \geq \frac{\vert \Lambda \vert}{2} \rho_z^2 \widehat g(0) + E_{d}^{\mathrm{LHY}}(\rho_{z}) +\mathcal{K}^{\rm{diag}} +\mathcal{Q}_{2}^{\mathrm{ex}} 
		+ (\rho_z -\rho)n_+ \widehat{g}(0) - \rho \rho_z |\Lambda| \widehat{g}(0)+ \rho^2|\Lambda|\widehat{g}(0) \nonumber\\
&= \frac{\vert \Lambda \vert}{2} \rho^2 \widehat g(0) + E_{d}^{\mathrm{LHY}}(\rho_{z}) +\mathcal{K}^{\rm{diag}} +\mathcal{Q}_{2}^{\mathrm{ex}} \nonumber
+ \frac 1 2 (\rho - \rho_z)^2 \vert \Lambda \vert \widehat g(0) + (\rho_z -\rho)n_+ \widehat{g}(0) .
	\end{align}
The last term we can bound by integrating and using \eqref{eq:condensationestimate},
\begin{equation}
\int_{\{\lvert \rho_z - \rho\rvert \leq \rho \varepsilon_+\}}  (\rho_z -\rho) \langle n_+ \rangle_{\Phi(z)} \widehat{g}(0) \, \dd z \leq C \rho \varepsilon_+ \widehat{g}(0) \langle n_+ \rangle_{\Psi} = o_{d}^{\text{LHY}},
\end{equation}
thanks to the choice of $\varepsilon_+$. Therefore, we deduce
\begin{align}
		&\int_{\{\lvert \rho_z - \rho\rvert \leq \rho \varepsilon_+\}}  \langle \mathcal{K}(z) \rangle_{\Phi(z)} \dd z  \geq \\
		&\int_{\{\lvert \rho_z - \rho\rvert \leq \rho \varepsilon_+\}} \Big(\frac{\vert \Lambda \vert}{2} \rho^2 \widehat g(0) + E_{d}^{\mathrm{LHY}}(\rho_{z})\Big)\|\Phi(z)\|^2 + \langle \mathcal{K}^{\rm{diag}} + \mathcal{Q}_{2}^{\mathrm{ex}}(z) \rangle_{\Phi(z)}\,\dd z + o^{\rm{LHY}}_d.\nonumber
\end{align}
 The contributions of $\mathcal{Q}_{3}^{\mathrm{soft}}$, $\mathcal{Q}_{2}^{\mathrm{ex}}$, and $\mathcal{K}^{\rm{diag}}$ are combined using Proposition~\ref{prop:Q3z}. Bounding the remaining positive terms by $0$ and estimating the errors with the relations from \ref{app:parameters}, we deduce
	\begin{multline}\label{lowbound_rhoclose}
		\int_{\{\lvert \rho_z - \rho\rvert \leq \rho \varepsilon_+\}} \langle \mathcal{K}(z)\rangle_{\Phi(z)} \,dz\\
		\geq\int_{\{\lvert \rho_z - \rho\rvert \leq \rho \varepsilon_+\}}   \Big(\frac{1}{2}\rho^2 |\Lambda| \widehat{g}(0) + E_{d}^{\mathrm{LHY}}(\rho_{z}) \Big) \|\Phi(z)\|^2 \dd z+ o^{\text{LHY}}_d. 
	\end{multline}
	Finally, in this case we can replace $\rho_z$ by $\rho$ up to errors of order $o^{\text{LHY}}_d$. Hence we have a lower bound for all $z$, and we deduce from \eqref{eq:intKz}, from the contributions of the integrals in \eqref{lowbound_rhofar} and \eqref{lowbound_rhoclose} on the domains $\{z \in \mathbb{C}: \lvert \rho_z - \rho\rvert \geq \varepsilon_+ \rho\}$ and $\{z \in \mathbb{C}: \lvert \rho_z - \rho\rvert < \varepsilon_+ \rho\}$, respectively, that
	\begin{equation}
		\langle \mathcal H \rangle_\Psi \geq \frac{\rho^2}{2} \vert \Lambda \vert \widehat g(0) + E_d^{\text{LHY}}(\rho) 
		+ o_d^{\rm{LHY}},
	\end{equation}
	which concludes the proof of Theorem~\ref{thm:main}.
\end{itemize}

\appendix

\section{Miscellaneous Estimates}

\begin{lemma}\label{lem:gomega.approx}
There exists a constant $C>0$ such that the following estimate holds
\begin{equation*}
\Big\vert \widehat{g\omega}(0) - \frac{1}{\vert \Lambda \vert} \sum_{k \in \mathcal P_H} \frac{\widehat g(k)^2}{2k^2} \Big\vert \leq C \widehat g(0) K_H^{-1} + \mathcal E_d,
\end{equation*}
where 
\begin{equation*}
\mathcal E_d \leq 
\begin{cases}
CR^2 \ell^{-2}_\delta \widehat g(0)^2 + C\widehat g(0)^2 \vert \log K_H \ell_\delta \ell^{-1} \vert, &\quad\text{if } d=2, \\
C \widehat g(0)^2 K_H \ell^{-1}, &\quad\text{if } d=3.
\end{cases}
\end{equation*}
\end{lemma}
The constant $C$ in the error bounds depends on $L^p$-properties of the potential, $p>1$.
\begin{proof}
First of all, one can replace the sum by an integral,
\begin{equation}\label{eq:BadEstimate}
\Big\vert \frac{1}{\vert \Lambda \vert} \sum_{k \in \mathcal P_H} \frac{\widehat g(k)^2}{2k^2} - \int_{\vert k \vert \geq K_H \ell^{-1}} \frac{\widehat g(k)^2}{2k^2} \frac{\dd k}{(2\pi)^d} \Big\vert \leq C \widehat g(0) K_H^{-1}.
\end{equation}
This can be proven by bounding the derivatives of the integrand on small boxes of size $(2\pi) \ell^{-1}$, but
depends on $L^p$-properties of the potential, since we need some decay of $\widehat g(k)$ to control the decay of the summand. The estimate is obtained through a H\"{o}lder inequality on the sum.

Now we can compare the integral with $\widehat{g\omega}(0)$ (in $d=3$ for instance),
\begin{align}
\Big\vert \widehat{g\omega}(0) - \int_{\vert k \vert \geq K_H \ell^{-1}} \frac{\widehat g(k)^2}{2k^2} \frac{\dd k}{(2\pi)^d} \Big\vert &\leq \Big\vert \int_{\vert k \vert \leq K_H \ell^{-1}} \frac{\widehat g(k)^2}{2k^2} \frac{\dd k}{(2\pi)^d} \Big\vert \nonumber \\
&\leq C \widehat g(0)^2 K_H \ell^{-1}.
\end{align}
The estimate is similar in $d=2$, except we must bound $\vert \widehat g_k - \widehat g(0) \vert \leq R^2 \widehat g(0) k^2$ for small $k$'s, to have integrability.
\end{proof}
We end this section by stating, without proof, the following simple bounds, which will be useful for further estimates.
\begin{lemma}\label{Lem:alphak}
	If $\vert \rho_z - \rho \vert \leq \frac{1}{2}\rho$ and $\vert k \vert \geq K_H \ell^{-1}$. Then
	\begin{equation} \label{eq:alpha_highcontrol}
		\vert \alpha_k \vert \leq C \frac{\vert \rho_z \widehat g(k) \vert}{k^2},\quad\text{and}\quad
		\vert \mathcal D_k - k^2\vert \leq C \ell^2 \rho \widehat g _0 K_H^{-2} k^2.
	\end{equation}
\end{lemma}

\section{Localization of Large Matrices: restrictions of $n_+^L$}\label{sec:largeM}

Some of our errors depend on $n_+^L$. Thus, we need a priori bounds on this excitation number, for low energy states. We explain how we can reduce the analysis to states with bounded number of low excitations, $n_+^L \leq \mathcal M$, in Proposition~\ref{thm:excitationrestriction}.

\begin{proposition}\label{thm:excitationrestriction}
	There exist $C$, $\eta > 0$ such that the following holds. Let $\Psi \in L^2_{\rm{sym}}(\Lambda^N)$ be a normalized $N$-particle state which satisfies
	\begin{equation}\label{eq:condensationaprioricondition}
		\langle \mathcal H \rangle_\Psi \leq \frac{1}{2} \rho^2\vert \Lambda \vert \widehat g(0) + C_B \rho^2\vert \Lambda \vert \widehat{g}(0) \f
	\end{equation}
	for some $C_{B}>0$.
	Assume that $\mathcal M$ and $\Vert v \Vert_1$ satisfy \eqref{eq:boundonM}. Then, there exists a sequence $\{\Psi^{m}\}_{m \in \mathbb{Z}} \subseteq L^2_{\rm{sym}}(\Lambda^N)$ such that $\sum_m \Vert \Psi^m \Vert^2 = 1$ and
	\begin{equation}\label{eq:large_matrices_restrictioncondition}
		\Psi^{m} = \one_{[0,\frac{\mathcal{M}}{2} +m]}(n_+^L) \Psi^{m},
	\end{equation}
	and such that the following lower bound holds true
	\begin{align*}
		\langle \Psi, \mathcal{H} \Psi \rangle  \geq&  \sum_{2 |m|\leq \mathcal{M} }\langle \Psi^{m}, \mathcal{H} \Psi^m\rangle +   \frac{\vert \Lambda \vert}{2} \rho^2 \widehat g(0)\Big(1 + 2C_B  \lambda_d^{\rm{LHY}} \Big)\sum_{2|m| >\mathcal{M}} \|\Psi^m\|^2 +  o^{\rm{LHY}}_d.
	\end{align*}
\end{proposition}
The proof of Proposition~\ref{thm:excitationrestriction} will follow from the Lemmas~\ref{lem:localization_largeMat} and~\ref{lem:d1d2estimate} below. The proof of Lemma~\ref{lem:localization_largeMat} is inspired by the localization of large matrices result in \cite{LS_2comp}. It is also similar to the bounds in \cite[Proposition 21]{HST}. It can be interpreted as an analogue of the standard IMS localization formula. The error produced is written in terms of the following quantities $d_1^L$ and $d_2^L$ ($\mathrm{i.e}$ the terms in the Hamiltonian that change $n_{+}^{L}$ by $1$ or $2$).
\begin{align}
d_{1}^L &:= \sum_{i \neq j} (P_i + Q_{H,i})\overline{Q}_{H,j} v(x_i-x_j) \overline{Q}_{H,i} \overline{Q}_{H,j} + h.c. \nonumber \\
&\quad + \sum_{i \neq j} \overline{Q}_{H,i} (P_j + Q_{H,j}) v(x_i-x_j) (P_i + Q_{H,i})(P_j + Q_{H,j}) + h.c. \label{def.d1L} 
\intertext{and}
d_2^L &:= \sum_{i \neq j} (P_i + Q_{H,i})(P_j + Q_{H,j}) v(x_i-x_j) \overline{Q}_{H,j} \overline{Q}_{H,i} +h.c. \label{def.d2L}
\end{align}
where $Q_{H,j}$ is defined in \eqref{def:QH}.
These error terms are estimated in Lemma~\ref{lem:d1d2estimate}.

\begin{lemma}\label{lem:localization_largeMat}
Let $\theta : \R \rightarrow [0,1]$ be any compactly supported Lipschitz function such that $\theta(s) = 1$ for $\vert s \vert < \frac 1 8$ and $\theta(s) = 0$ for $\vert s \vert > \frac 1 4$. For any $\mathcal M >0$, define $c_{\mathcal M} >0$  and $\theta_{\mathcal M}$ such that
\[ \theta_{\mathcal M}(s) = c_{\mathcal M} \theta \Big( \frac{s}{\mathcal M} \Big) , \qquad \sum_{s \in \Z} \theta_{\mathcal M}(s)^2 = 1 .\]
Then there exists a $C>0$ depending only on $\theta$ such that, for any normalized state $\Psi \in L^2_{\rm{sym}}(\Lambda^N)$,
\begin{equation}\label{eq:large_matrices_decomposn+}
 \langle \Psi , \mathcal H \Psi \rangle \geq \sum_{m \in \Z} \langle \Psi^m , \mathcal H \Psi^m \rangle - \frac{C}{\mathcal{M}^2} \left( \vert \langle d_1^L \rangle_\Psi \vert + \vert \langle d_2^L \rangle_\Psi \vert \right),
 \end{equation}
where $\Psi^m = \theta_{\mathcal M}(n_+^L - m) \Psi$.
\end{lemma}

\begin{proof}
Notice that $\mathcal{H}$ only contains terms that change $n_+^L$ by $0, \pm 1$ or $\pm 2$. Therefore, we write our operator as $ \mathcal H = \sum_{\vert k \vert \leq 2} \mathcal{H}^{(k)}$,
with $ \mathcal H^{(k)} n_+^L = (n_+^L + k) \mathcal H^{(k)}$. Moreover, $\mathcal H^{(k)} + \mathcal H^{(-k)} = d^L_k$ for $k=1,2$. We use this decomposition to estimate the localized energy,
\begin{align*}
\sum_{m \in \Z} \langle \Psi^m, \mathcal H \Psi^m \rangle &= \sum_{m \in \Z} \sum_{\vert k \vert \leq 2} \langle \theta_{\mathcal M}(n_+^L -m)  \theta_{\mathcal M}(n_+^L-m+k) \Psi, \mathcal H^{(k)} \Psi \rangle\\
&= \sum_{m, s \in \Z} \sum_{\vert k \vert \leq 2} \langle \theta_{\mathcal M}(s-m)  \theta_{\mathcal M}(s-m+k) \one_{\{n_+^L =s\}} \Psi, \mathcal{H}^{(k)} \Psi \rangle\\
&= \sum_{m,s \in \Z} \sum_{\vert k \vert \leq 2} \theta_{\mathcal M}(m) \theta_{\mathcal M}(m+k) \langle \one_{\{n_+^L = s\}} \Psi, \mathcal H^{(k)} \Psi \rangle,
\end{align*}
where in the last line we changed the index $m$ into $s-m$. We can sum on $s$ to recognize
\begin{equation}
\sum_{m \in \Z} \langle \Psi^m, \mathcal H \Psi^m \rangle = \sum_{m \in \Z} \sum_{\vert k \vert \leq 2} \theta_{\mathcal M}(m) \theta_{\mathcal M}(m+k) \langle \Psi, \mathcal H^{(k)}  \Psi \rangle.
\end{equation}
Furthermore the energy of $\Psi$ can be rewritten as
\begin{equation}
\langle \Psi, \mathcal H \Psi \rangle = \sum_{\vert k \vert \leq 2} \langle \Psi, \mathcal H^{(k)}  \Psi \rangle = \sum_{m \in \Z} \sum_{\vert k \vert \leq 2} \theta_{\mathcal M}(m)^2 \langle \Psi, \mathcal H^{(k)} \Psi \rangle,
\end{equation}
by definition of $\theta_{\mathcal M}$. Thus, the localization error is
\begin{equation}
\sum_{m \in \Z} \langle \Psi^m, \mathcal H \Psi^m \rangle - \langle \Psi, \mathcal H \Psi \rangle = \sum_{\vert k \vert \leq 2} \delta_k \langle \Psi, \mathcal H^{(k)} \Psi \rangle,
\end{equation}
with 
\begin{equation}\label{eq.def.deltak}
\delta_k = \sum_{m \in \Z} \big (  \theta_{\mathcal M} (m)  \theta_{\mathcal M}(m+k) -  \theta_{\mathcal M}(m)^2  \big ) = - \frac{1}{2} \sum_m \big ( \theta_{\mathcal M}(m) -  \theta_{\mathcal M}(m+k) \big )^2.
\end{equation}
Since $\delta_0 = 0$, $\delta_k = \delta_{-k}$ and $d_k^L = \mathcal H^{(k)} + \mathcal H^{(-k)}$ we find
\begin{equation}
\sum_{m \in \Z} \langle \Psi^m, \mathcal H \Psi^m \rangle - \langle \Psi, \mathcal H \Psi \rangle = \delta_1 \langle d_1^L \rangle_\Psi + \delta_2 \langle d_2^L \rangle_\Psi,
\end{equation}
and only remains to prove that $\vert \delta_k \vert \leq C \mathcal M^{-2}$. This follows from \eqref{eq.def.deltak} using that $\theta$ is Lipschitz and restricting the sum to $ m  \in \big[- \frac{ \mathcal M}{2} , \frac{ \mathcal M}{2} \big]$.

\end{proof}

To estimate the error in \eqref{eq:large_matrices_decomposn+}, we need the following bounds on $d_1^L$ and $d_2^L$.
\begin{lemma}\label{lem:d1d2estimate}
There exists a universal constant $C > 0$ such that, for any $\Psi \in L_{\rm{sym}}^2(\Lambda^N)$, with our choices of parameters we have
\begin{align}\label{eq:estimateMd1d2}
\vert \langle d_{1}^L \rangle_\Psi \vert + \vert \langle d_2^L\rangle_{\Psi} \vert  \leq C \|v\|_1 \rho K_H \langle n_+\rangle_{\Psi} +C \langle \mathcal Q_4^{\rm{ren}}\rangle_{\Psi}.
\end{align}
\end{lemma}

\begin{proof}
First note that we have the following bound on the operator norm
\begin{equation}\label{eq:boundQH}
\Vert \overline{Q}_{H,x} v(x - y) \overline Q_{H,x} \Vert \leq C K_H^2 \ell^{-d} \Vert v \Vert_1.
\end{equation}
Indeed, for all $\varphi \in \mathrm{Ran} \overline {Q}_{H,x}$,
\begin{equation}
\langle \overline{Q}_{H,x} v(x - y) \overline Q_{H,x} \varphi, \varphi \rangle \leq \int_\Lambda \vert \varphi(x) \vert^2 v(x - y) \dd x \leq \Vert \varphi \Vert_{\infty}^2 \Vert v \Vert_1 \leq C \ell^{2-d} \Vert \Delta \varphi \Vert \Vert \varphi \Vert \Vert v \Vert_1,
\end{equation}
by Sobolev inequality. Moreover such $\varphi$'s satisfy $\Vert \Delta \varphi \Vert \leq K_H^2 \ell^{-2} \Vert \varphi \Vert$ by definition of $\overline Q_H$, and \eqref{eq:boundQH} follows.

We split $d_{1}^L$, $d_2^L$ in several terms multiplying out the parentheses in \eqref{def.d1L} and \eqref{def.d2L}. Here we just bound some representative examples to illustrate the procedure.

For instance, we can use the Cauchy-Schwarz inequality  with weight $K_H$ and equation \eqref{eq:boundQH} to find,
\begin{align*}
\Big|\Big\langle  \sum_{i,j} P_i \overline{Q}_{H,j} v \overline{Q}_{H,i} \overline{Q}_{H,j}\Big\rangle_{\Psi}\Big| &\leq  K_H \frac{N}{\vert \Lambda \vert} \|v\|_1\langle n^L_+ \rangle_{\Psi}  + K_H^{-1} \|\overline{Q}_H v\overline{Q}_H\| N \langle n_+^L \rangle_{\Psi},\\
&\leq C \|v\|_1 K_H \rho \langle n_+ \rangle_\Psi
\end{align*}
where we used $n_+^L \leq n_+$.

We also estimate a term where the need for $\mathcal{Q}_4^{\text{ren}}$ becomes clear. In order to do that we complete the $Q_H$ to a $Q= Q_H + \overline{Q}_H$, 
\begin{align*}
&\Big|\Big\langle  \sum_{i \neq j} \overline{Q}_{H,i} P_j v Q_{H,i} Q_{H,j}  + h.c.\Big\rangle_{\Psi} \Big| \\
&\quad\leq   \Big|\Big\langle \sum_{i \neq j} \overline{Q}_{H,i} P_j v (Q_{H,i} \overline{Q}_{H,j} + \overline{Q}_{H,i} Q_{H,j}) \Big\rangle_{\Psi} + h.c. \Big| \\
&\qquad +\Big|\Big\langle \sum_{i \neq j} \overline{Q}_{H,i} P_j v Q_i Q_j \Big\rangle_{\Psi} + h.c. \Big|  
+ \Big|\Big\langle \sum_{i,j} P_i \overline{Q}_{H,j}v \overline{Q}_{H,i} \overline{Q}_{H,j} \Big\rangle_{\Psi}\Big|.
\end{align*}
The first and the third terms can be estimated in the same manner as above, so let us focus on completing the second term in order to obtain $4Q$ terms.
\begin{align}
 \Big|\Big\langle&\sum_{i \neq j}  \overline{Q}_{H,i} P_j v Q_i Q_j \Big\rangle_{\Psi} + h.c. \Big|  \\
 &\leq \Big|\Big\langle \sum_{i \neq j} \overline{Q}_{H,i} P_j v (Q_i Q_j + \omega (P_i P_j + P_i Q_j + Q_i P_j)) \Big\rangle_{\Psi} + h.c. \Big|  \label{eq:Q4recostructionQhigh}\\
 &\quad + \Big|\Big\langle \sum_{i \neq j} \overline{Q}_{H,i} P_j v  \omega ( P_i Q_j + Q_i P_j)) \Big\rangle_{\Psi} + h.c. \Big|  \\
 &\quad + \Big|\Big\langle \sum_{i \neq j} \overline{Q}_{H,i} P_j v  \omega P_i P_j) \Big\rangle_{\Psi} + h.c. \Big|.
\end{align}
The second and the third terms are treated as above, using that $0 \leq \omega \leq 1$ on the support of $v$. By a Cauchy-Schwarz inequality on the first term we get
\begin{equation*}
\eqref{eq:Q4recostructionQhigh} \leq \langle \mathcal Q_4^{\text{ren}}\rangle_{\Psi} + C \frac{N}{\vert \Lambda \vert} \|v\|_1 \langle n_+\rangle_{\Psi}.
\end{equation*}
\end{proof}

Now we can combine Lemmas~\ref{lem:localization_largeMat} and~\ref{lem:d1d2estimate} to prove Proposition~\ref{thm:excitationrestriction}.
\begin{proof}[Proposition~\ref{thm:excitationrestriction}]
Given $\Psi \in L_{\mathrm{sym}}^2(\Lambda^N)$ satisfying \eqref{eq:condensationaprioricondition}, we can apply Lemma~\ref{lem:localization_largeMat} and write $\Psi^m = \theta_{\mathcal M}(n_+^L -m)\Psi$. In \eqref{eq:large_matrices_decomposn+} we split the sum into two.
The first part, for $|m| < \frac 1 2 \mathcal M$, we keep. For $|m| > \frac 1 2 \mathcal{M}$, $\Psi_m$ satisfies
\begin{equation}\label{eq:M}
\langle n_+\rangle_{\Psi^m} \geq \langle n_+^L\rangle_{\Psi^m} \geq \frac{ \mathcal{M}}{4}\| \Psi^m \|^2 ,  
\end{equation}
due to the cutoff $\theta_{\mathcal M}(n_+^L - m)$. Since we have from \eqref{eq:boundonM1} that $\mathcal M \gg \rho^2\ell^2 \vert \Lambda \vert \widehat{g}(0) \f$, this is a larger bound than \eqref{eq:condensationestimate}, and thus the assumption of Theorem~\ref{thm:apriori_n+} cannot be satisfied for $\Psi^m$ and we must have the lower bound 
\begin{equation}\label{eq:PsimH.01}
\langle \Psi^m, \mathcal{H} \Psi^m\rangle \geq \rho^2\vert \Lambda \vert\widehat{g}(0)\Big(  \frac{1}{2} + C_B  \f  \Big) \|\Psi^m \|^2.
\end{equation}
We finally bound the last term in \eqref{eq:large_matrices_decomposn+}, using Lemma~\ref{lem:d1d2estimate}. We use the condensation estimate \eqref{eq:condensationestimate} and the bound \eqref{eq:estimateprioriQ4} on $\mathcal Q_4^{\rm{ren}}$ to obtain 
\begin{align}\label{eq:d1d2.01}
\mathcal M^{-2} \big( \vert \langle d_{1}^L \rangle_\Psi \vert + \vert \langle d_2^L\rangle_{\Psi} \vert \big) &\leq C \mathcal M^{-2} \Big( \rho K_H \Vert v \Vert_1  \ell^2 + 1 \Big)\vert \Lambda \vert\rho^2 \widehat{g}(0)\f \nonumber \\
&=  o^{\rm{LHY}}_d,
\end{align}
for $\mathcal M$ and $\Vert v \Vert_1$ satisfying \eqref{eq:boundonM}. Using the estimates \eqref{eq:PsimH.01} for $m > \frac 1 2 \mathcal M$ and \eqref{eq:d1d2.01} in formula \eqref{eq:large_matrices_decomposn+} we conclude the proof.
\end{proof}

\section{Rigorous Bogoliubov Theory for Quadratic Hamiltonians}

\subsection{Diagonalization of quadratic Hamiltonians}

In the next proposition we show a simple consequence of the Bogoliubov method, see \cite[Theorem 6.3]{Lieb_2001} and \cite{brietzke_second-order_2020}, that we use to diagonalize the quadratic term $\mathcal Q(z)$ of Proposition~\ref{prop:projonz}. 
\begin{theorem}\label{thm:bogdiag}
	Let $a_{\pm}$ be operators on a Hilbert space satisfying $\left[ a_+, a_- \right] = 0$. For $\mathcal A > 0$, $\mathcal B \in \R$ satisfying $\vert \mathcal B \vert < \mathcal A$ and arbitrary $\kappa \in \mathbb{C}$, we have the operator identity
	\begin{align*}
		\mathcal{A}&(a^{\dagger}_+ a_+ + a^{\dagger}_- a_-) + \mathcal{B}(a_+^{\dagger} a^{\dagger}_- + a_+ a_-) + \kappa (a^{\dagger}_+ + a_-) + \overline{\kappa}(a_+ + a_-^{\dagger})  \\
		&=  \mathcal D ( b^\dagger_+ b_+ + b^\dagger_- b_- ) -\frac{1}{2}\alpha\mathcal{B} ([a_+,a^{\dagger}_+] + [a_-, a^{\dagger}_-]) - \frac{2|\kappa|^2}{\mathcal{A}+\mathcal{B}},
	\end{align*}
	where $\mathcal D = \frac{1}{2} \Big( \mathcal{A} + \sqrt{\mathcal A^2 - \mathcal B^2}\Big)$, and
	\begin{equation}
		b_+ = \frac{1}{\sqrt{1 - \alpha^2}} \big( a_+ + \alpha a_-^\dagger + \bar c_0 \big), \qquad b_- =  \frac{1}{\sqrt{1 - \alpha^2}} \big( a_- + \alpha a_+^\dagger + c_0 \big),
	\end{equation}
	with 
	\begin{equation}
		\alpha = \mathcal B^{-1} \big( \mathcal A - \sqrt{\mathcal A^2 - \mathcal B^2} \big), \qquad c_0 = \frac{2 \bar \kappa}{\mathcal A + \mathcal B + \sqrt{\mathcal A^2 - \mathcal B^2}}.
	\end{equation}
\end{theorem}

\begin{remark}
	Note that the normalization of $b_\pm$ is chosen such that 
	\begin{equation}
		[ b_+, b_+^\dagger ] = \frac{ [ a_+ , a_+^\dagger ] - \alpha^2 [ a_- , a_-^\dagger ]}{1 - \alpha^2} , 
	\end{equation}
	and we recover the canonical commutation relations $[ b_+, b_+^\dagger ] = 1$ when $a_+$ and $a_-$ satisfies them as well.
\end{remark}

\begin{proof}
	This follows directly from algebraic computations.
\end{proof}

\subsection{Evaluation of the Bogoliubov integral}

In this section we report two lemmas for the calculation of the Bogoliubov integral. The first one, under weak assumptions, gives a bound for general Bogoliubov-type integrals, expressing the dependence on the parameters involved in the spectral gaps. The second one is a more precise calculation which lets us obtain the exact value of the Lee-Huang-Yang constant.
Let us recall the definition of $G_d$ in \eqref{defG}:
\begin{equation}
G_d(k) := \frac{\widehat{g}_{\mathbb{R}^d}(k)^2 - \widehat{g}_{\mathbb{R}^d}(0)^2 \one_d(\ell_{\delta}k)}{2 k^2}.
\end{equation}

\begin{lemma}\label{lem:Bog_int_parameters}
Let $\mathcal{A}, \mathcal{B}: \mathbb{R}^d \rightarrow \mathbb{R}$ be two functions such that, for parameters satisfying $\kappa >0$, $0< K_2 \leq K_1$, $\ell_{\delta}^{-1}\leq  K < a^{-1}$, 
\begin{align}
\mathcal{A}(k) \geq \kappa[|k|&- K]_+^2 +2  K_1 \widehat{g}(0), \qquad |\mathcal{B}(k)| \leq 2 K_2 \widehat{g}(0),\nonumber\\
&\lvert \mathcal{B}(k) -\mathcal{B}(0)\rvert \leq K_2 R^2 \widehat{g}(0)|k|^2,
\end{align}
and let us introduce the integral, recalling \eqref{defG}, 
\begin{equation}
I(d) = \int_{\mathbb{R}^d} \Big( \mathcal{A}(k) - \sqrt{\mathcal{A}(k)^2 - \mathcal{B}(k)^2}\Big) \dd k-\frac{K_2^2}{\kappa}\int_{\mathbb{R}^d} G_d
(k) \dd k,
\end{equation}
then there exists a constant $C>0$ such that
\begin{itemize}
\item  For $d=3$, 
\begin{align*}
 I(3) &\leq C \frac{K K_2^2a}{\kappa} \widehat{g\omega}(0) + C\widehat{g}(0) K_2^2\big( K_1^{-1}K^3 +\kappa^{-1} \widehat{g}(0) K \log((aK)^{-1}) \big)\\&\quad+\min \Big(\kappa^{-3} \widehat{g}(0)^4 \frac{K_2^4}{K^{3}},\frac{K_2^4}{K_1^2} \widehat{g\omega}(0)\Big).
\end{align*}
\item  For $d=2$, 
\begin{align*}
 I(2)&\leq  C \widehat{g}(0) K_2^2\Big( \widehat{g}(0)(\rho K_1^{-1} + \kappa^{-1}  R^2\ell_{\delta}^{-2}) + \kappa^{-1} \widehat{g}(0)|\log(2K\ell_{\delta})| + \kappa^{-1}\widehat{g}(0)\Big) \\&\quad+ \min \Big(\kappa^{-3} \widehat{g}(0)^4 \frac{K_2^4}{K^{4}}, \frac{K_2^4}{K_1^2} \widehat{g\omega}(0)\Big).
\end{align*}
\end{itemize}
\end{lemma}

\begin{proof}
The proof of the 3D and 2D cases can be found in \cite[Lemma C.1]{FS2} and \cite[Lemma C.5]{2DLHY}, respectively.
\end{proof}

\begin{lemma}\label{lem:calculation_LHYterm_int}
There exists a $C>0$ such that 
\begin{equation}
\frac{1}{2(2\pi)^d} \int_{\mathbb{R}^d} \Big( \sqrt{k^4 - 2k^2 \rho \widehat{g}(k)} -k^2 \rho \widehat{g}(k)- \rho^2 G_d(k) \Big)\dd k 
=\frac{\rho^2}{2} I^{\text{Bog}}_d \widehat{g}(0)\lambda_d^{\text{LHY}} + \mathcal{E}_d^{\text{int}}(\rho),
\end{equation}
where
\begin{equation}
|\mathcal{E}_d^{\text{int}}(\rho)| \leq \begin{dcases} C \rho^2 \widehat{g}(0)^3 \rho R^2 \log(\widehat{g}(0)), \quad &\text{if } d=2,\\
C\rho^2 \widehat{g}(0)^3 \rho R^2 \sqrt{\rho \widehat{g}(0)^3}, \quad &\text{if }  d=3.
\end{dcases}
\end{equation}
\end{lemma}
\begin{proof}
The idea of the proof is to estimate the error made approximating $\widehat{g}(k)$ with $\widehat{g}(0)$ and then changing variables $k \mapsto \sqrt{\rho \widehat{g}(0)} k$ to reduce to $I_d^{\text{Bog}}$. The details can be found in \cite[Proposition C.3]{2DLHY} and \cite[Lemma C.2]{FS2} for dimension 2 and 3, respectively.
\end{proof}

\section{When $\rho_z$ is far from $\rho$}\label{app:rho_far}

Before establishing the lower bound when $\vert \rho -\rho_{z}\vert\geq \rho \varepsilon_+$, we first need the following intermediate lemma, which states that the elements corresponding to the soft pairs interaction in $\mathcal{Q}_3^{\text{ren}}$ can be bounded at the price of a small part of the kinetic energy. We recall the definition of $\mathcal{Q}_3^{\text{soft}}$ in \eqref{def:Q3soft} and the definition of the momenta spaces $\mathcal{P}_L$ and $\mathcal{P}_H$ in \eqref{def:PLPH}.

\begin{lemma}\label{lemma:Q3} There exists a universal constant $C>0$ such that, for any $z \in \mathbb C$, any $\varepsilon >0$, and any $\Phi \in \mathcal F_s(\rm{Ran}Q)$ satisfying
\begin{equation}
\langle n_+ \rangle_\Phi \leq \rho\vert \Lambda \vert ,\label{eq:phi}
\end{equation}
we have
\begin{align}\label{ine:Q3}
\Big\langle \frac{\varepsilon}{2}\sum_{k\in\mathcal P_H }k^{2}a_{k}^{\dagger}a_{k}+\mathcal Q_3^{\rm{soft}}(z) \Big\rangle_{\Phi}\geq - C\vert \Lambda\vert  \varepsilon^{-1} \rho \rho_z \widehat g(0) \frac{K_{\ell}^2}{ K_{H}^{2}}\frac{\langle n_+^L\rangle_{\Phi}}{N} K_{L}^{d}.
\end{align}
\end{lemma}

\begin{proof}
Introducing the operators
\begin{equation}
b_k := a_k + \frac{2}{\varepsilon \vert\Lambda\vert} \sum_{p \in \mathcal{P}_L} \, \frac{\widehat{g}(k)}{k^2} z a_{p-k}^{\dagger} a_p,
\end{equation}
and 
\begin{equation}
	\mathcal{K}^{\mathrm{diag}}_{\varepsilon}=\frac{\varepsilon}{2}\sum_{k\in\mathcal P_H }k^{2}a_{k}^{\dagger}a_{k}
\end{equation}
we can complete the square in the following expression, obtaining
\begin{align*}
\mathcal{K}^{\mathrm{diag}}_{\varepsilon}+\mathcal Q_3^{\rm{soft}} &= \sum_{k \in \mathcal{P}_H} \, \Big( \frac{\varepsilon}{2} k^2 b^{\dagger}_k b_k - \frac{2|z|^2}{\varepsilon\vert\Lambda\vert^{2}} \sum_{p,s \in \mathcal{P}_L}  \frac{\widehat{g}(k)^2}{k^2} a_s^{\dagger} a_{s-k} a^{\dagger}_{p-k}a_p\Big) \\
&\geq - \frac{2|z|^2}{\varepsilon\vert\Lambda\vert^{2}} \sum_{k \in \mathcal{P}_H} \, \frac{\widehat{g}(k)^2}{k^2} \sum_{p,s \in \mathcal{P}_L}  a_s^{\dagger}\big(a^{\dagger}_{p-k} a_{s-k} + [a_{s-k}, a^{\dagger}_{p-k}] \big)a_p.
\end{align*}
For the term without commutator, estimated on a state $\Phi$ which satisfies \eqref{eq:phi} and using the Cauchy-Schwarz inequality
\begin{equation}
a_s^{\dagger}a^{\dagger}_{p-k} a_{s-k}a_p \leq C (a_s^{\dagger}a^{\dagger}_{p-k} a_{p-k}a_s + a_p^{\dagger}a^{\dagger}_{s-k} a_{s-k}a_p)
\end{equation}
we have
\begin{align}
&\frac{2|z|^2}{\varepsilon\vert\Lambda\vert^{2}} \left\langle \sum_{k \in \mathcal{P}_H} \frac{\widehat{g}(k)^2}{k^2} \sum_{p,s\in \mathcal{P}_L} a_s^{\dagger}a^{\dagger}_{p-k} a_{p-k}a_s \right\rangle_{\Phi} \nonumber \\
&\leq C\varepsilon^{-1}\frac{\rho_{z}\widehat{g}(0)^{2}}{\vert \Lambda \vert} \sum_{k \in \frac{1}{2}\mathcal{P}_H} \sum_{s\in \mathcal{P}_L} \frac{1}{k^2} \langle a_s^{\dagger}a^{\dagger}_{k} a_{k}a_s \rangle_{\Phi} \Big (\sum_{p \in \mathcal{P}_L} 1 \Big ) \nonumber \\
&\leq C \varepsilon^{-1} \rho \rho_z \ell^2 \widehat{g}(0)^2 \frac{ \langle n_+^L \rangle_{\Phi} K_{L}^{d} }{K_{H}^{2}}, \label{eq:proofAtilde3}
\end{align}
where in the last line we used that the sum over $\frac{1}{2}\mathcal{P}_H$ of $a^{\dagger}_k a_k$ can be bounded by the number of bosons $N =\rho \lvert \Lambda\rvert$, while the sum over $\mathcal{P}_L$ of the $a^{\dagger}_s a_s$ can be bounded by $C\langle n_+^L \rangle_{\Phi}$ thanks to the assumptions on $\Phi$.

On the other hand, the commutator satisfies $a_s^{\dagger} [a_{s-k}, a^{\dagger}_{p-k}] a_p = \delta_{s=p} a_p^{\dagger}a_p $, so we get
\begin{multline}
\frac{2|z|^2}{\varepsilon\vert\Lambda\vert^{2}}\left\langle \sum_{k \in \mathcal{P}_H} \frac{\widehat{g}(k)^2}{k^2} \sum_{p,s\in \mathcal{P}_L} a_s^{\dagger} [a_{s-k}, a^{\dagger}_{p-k}] a_p \right\rangle_{\Phi} \\
\leq C\frac{|z|^2}{\varepsilon\vert\Lambda\vert^2} \sum_{k \in \mathcal P_H, p \in \mathcal P_L}\frac{\widehat g(k)^2}{k^2} \langle a_p^{\dagger} a_p \rangle_{\Phi} \leq C\varepsilon^{-1} \rho_z \widehat g(0) \langle n_+^L \rangle_{\Phi}, \label{eq:proofAtilde4}
\end{multline}
where we used Lemma~\ref{lem:gomega.approx}, and we obtain a term which is smaller than the error stated in the lemma provided $\frac{K_{\ell}^2 K_L^d}{K_H^2}\leq 1$. 

Combining the inequalities from \eqref{eq:proofAtilde3} and \eqref{eq:proofAtilde4} we get the estimate of the lemma.
\end{proof}

We are now ready to state the theorem which gives a lower bound for the expression \eqref{eq:c-numberHamiltonian} when $\vert \rho - \rho_z \vert \geq \rho \varepsilon_+$. 
We use the notation 
\begin{equation}
\Phi(z) := \langle z\vert \Psi \rangle, \qquad z \in \mathbb{C},
\end{equation}
where $\vert z \rangle$ belongs to the family of coherent states of the form \eqref{eq:coherent_z}, so that, from the c-number substitution, we can write
\begin{equation}\label{integral_expression_rhofar}
\langle \Psi,\mathcal{H} \Psi\rangle = \frac{1}{\pi} \int_{\mathbb{C}} \langle \Phi(z),(\mathcal{K}(z) + \mathcal{Q}_3^{\text{ren}}+ \mathcal{Q}_4^{\text{ren}} + \mathcal{R}_0)\Phi(z)\rangle_{} \mathrm{d}z.
\end{equation}
We further observe that, since $\Psi = \one_{[0,\mathcal{M}]}(n_+^L) \Psi$, we have
\begin{equation}\label{eq:numberPhiPsi}
\langle n_+^L\rangle_{\Phi(z)} \leq \mathcal{M} \lVert \Phi(z)\rVert^2,
\end{equation}
and the simpler
\begin{equation}
\langle n_+\rangle_{\Phi(z)} \leq N \lVert \Phi(z)\rVert^2.
\end{equation}

\begin{theorem}\label{thm:rhofar}
Assume $\vert \rho - \rho_z \vert \geq \rho \varepsilon_+$ and that the relations between the parameters in \ref{app:parameters} hold true. If there exists a $C>0$ such that $\rho a^d \leq C^{-1}$, then for any normalized, $N-$particle state $\Psi \in \mathscr{F}_s(L^2(\Lambda))$ satisfying \eqref{eq:assumption_lowE_psi} and $\Psi = \one_{[0,2\mathcal{M}]}(n_+^L)\Psi$,
the following lower bound holds,
\begin{equation*}
\langle \mathcal{K}(z) + \mathcal{Q}_3^{\rm{soft}}(z) +\mathcal{R}_0\rangle_{\Phi(z)} \\
\geq  \Big(\frac{1}{2} \rho^2 |\Lambda| \widehat{g}(0) + 2 E^{\rm{LHY}}_d\Big)\lVert \Phi(z)\rVert^2 - C\rho \widehat{g}(0) \langle n_+^H\rangle_{\Phi(z)}.
\end{equation*}
\end{theorem}

\begin{proof}
We start by proving the following lower bound 
\begin{align}\label{eq:intermediate_rough_lb}
&\langle \mathcal K(z) + \mathcal{Q}_3^{\text{soft}}\rangle_{\Phi(z)}  \nonumber \\
&\geq|\Lambda| \widehat{g}(0) \Big( \frac{1}{2}\rho^2_z+ \rho^2 - \rho \rho_z- C K_{\ell}^2 K_L^{-1} ( \rho \rho_z + \rho_z^2+ \rho^2 )
- C \rho^2   \f \Big) \lVert \Phi(z)\rVert^2\nonumber \\
&\quad-C\rho \widehat{g}(0) \langle n_+^H\rangle_{\Phi(z)} .
\end{align}
We use Lemma~\ref{lemma:Q3}. Subtracting a small part of the kinetic energy from $\mathcal K(z)$, we get a bound on $\mathcal Q_3^{\rm{soft}}(z)$,
\begin{align}
 \frac{\varepsilon}{2\pi}\Big\langle \sum_{k\in\mathcal P_H }k^{2}a_{k}^{\dagger}a_{k}+\mathcal Q_3^{\rm{soft}}(z) \Big\rangle_{\Phi(z)} &\geq - C\vert \Lambda \vert  \varepsilon^{-1} \rho \rho_z \widehat g(0) \frac{K_{\ell}^2}{ K_{H}^{2}}\frac{\langle n_+^L\rangle_{\Phi(z)}}{N} K_{L}^{d}
\nonumber \\
&\geq -C\vert \Lambda \vert \varepsilon^{-1} \rho \rho_z \widehat{g}(0) \frac{K_{\ell}^6}{K_L^3 }\lVert \Phi(z) \rVert^2, \label{eq:Q3estim_rhofar}
\end{align}
where we used \eqref{eq:numberPhiPsi} and the assumption on $\Psi$ to have $\langle n_+^L \rangle_{\Phi(z)} \leq C \mathcal{M} \lVert \Phi(z)\rVert^2$ and the relations between the parameters. Choosing 
\begin{equation}\label{eq:choice_of_eps}
\varepsilon = \frac{K_{\ell}^4}{K_L^2} \ll 1,
\end{equation}
this term can be absorbed in the $K_{\ell}^2 K_L^{-1}$ term in \eqref{eq:intermediate_rough_lb}.

 Subtracting $\varepsilon/2 \sum k^2 a^{\dagger}_k a_k$ from $\mathcal K^{\text{Bog}}$, for $\varepsilon \ll 1$, 
 this is turned into
\begin{equation}\label{KBogtilde}
\widetilde{\mathcal{K}}^{\text{Bog}} =  \frac{1}{2} \sum_{k \neq 0} \widetilde{\mathcal{A}}_k  \big( a_k^\dagger a_k + a_{-k}^\dagger a_{-k} \big)
  + \frac{1}{2} \sum_{k \neq 0} \mathcal{B}_k \big( a_k^\dagger a_{-k}^\dagger + a_k a_{-k}\big),
\end{equation}
where
\begin{equation}
\widetilde{\mathcal{A}}_k:= (1-\varepsilon)k^2  + \rho_z \widehat{g}_k.
\end{equation}

The diagonalization procedure in Proposition~\ref{prop:BogDiag} can be adapted with the modified kinetic energy, and we find 
\begin{align}
\widetilde{\mathcal{K}}^{\text{Bog}} &\geq -\frac{1}{2} \sum_{k \neq 0} \Big(\widetilde{\mathcal{A}}_k - \sqrt{\widetilde{\mathcal{A}}_k^2 -\mathcal{B}_k^2} \Big)   \nonumber\\
&\geq -\frac{|\Lambda|}{2(2\pi)^2} \int_{\mathbb{R}^d} \Big(\widetilde{\mathcal{A}}_k - \sqrt{\widetilde{\mathcal{A}}_k^2 -\mathcal{B}_k^2} \Big) \mathrm{d}k  + o^{\text{LHY}}_d, 
\end{align}
where we approximated the series by the integral obtaining a small error absorbed in the last term. Since 
\begin{equation}
\widetilde{\mathcal{A}}_k \geq (1-\varepsilon) \Big[|k| - \sqrt{\rho \widehat{g}(0)} \Big]_+^2 + \frac{1}{2}\rho_z \widehat{g}(0),
\end{equation}
we satisfy the assumptions of Lemma~\ref{lem:Bog_int_parameters}, with $\kappa = (1-\varepsilon), K = \sqrt{\rho \widehat{g}(0)}$, $K_1 =  \frac{1}{2}\rho_z$, $K_2 = \rho_z$, and therefore we get the estimate 
\begin{align} 
&\frac{1}{2}\rho_z^2|\Lambda|\widehat{g \omega}(0) -\frac{|\Lambda|}{2(2\pi)^2} \int_{\mathbb{R}^d} \Big(\widetilde{\mathcal{A}}_k - \sqrt{\widetilde{\mathcal{A}}_k^2 -\mathcal{B}_k^2} \Big) \mathrm{d}k \nonumber \\
&\geq -C \varepsilon\rho_z^2 |\Lambda| \widehat{g \omega}(0)- C \rho \rho_z |\Lambda| \widehat{g}(0) \f \nonumber\\
&\quad - C \rho_z^2 |\Lambda| \widehat{g}(0)(1-\varepsilon)^{-1}(\f + R^2 \ell_{\delta}^{-2}\one_{d,2} ) +  o^{\text{LHY}}_d \nonumber \\
&\geq - C \rho_z^2 |\Lambda| \widehat{g}(0) (\varepsilon + R^2 \ell_{\delta}^{-2}\one_{d,2} + \f)- C \rho^2 |\Lambda| \widehat{g}(0) \f, \label{eq:int_estimate_rhofar}
\end{align}
where we reconstructed $\widehat{g\omega}(0)$ obtaining an error reabsorbed in the first term of the third line, and we used a Cauchy-Schwarz inequality on the second term in the second line.
Thanks to the choice of $\varepsilon$ made in \eqref{eq:choice_of_eps} and the relations between the parameters, we have that $\varepsilon$ is the dominant term in the first addend, and it can be reabsorbed in the $K_{\ell}^2 K_L^{-1}$ term in \eqref{eq:intermediate_rough_lb}, while the second addend is dominated by error term in \eqref{eq:intermediate_rough_lb}.

We bound by zero the positive terms in the quadratic elements in creation and annihilation operators 
\begin{align}
\langle(\rho_z-\rho) n_+  \widehat g(0) + \mathcal Q_2^{\rm{ex}}(z) \rangle_{\Phi(z)} &\geq -\rho \widehat{g}(0)\langle n_+ \rangle_{\Phi(z)} \nonumber \\
&\geq -C \rho \widehat{g}(0) \big( \mathcal{M}\lVert \Phi(z) \rVert^2  + \langle n_+^H\rangle_{\Phi(z)}\big),
\label{ine:Q2} 
\end{align}
where we used the simple bound $n_+ \leq C (n_+^L + n_+^H)$ and \eqref{eq:numberPhiPsi}.
The first term, thanks to \eqref{eq:boundonM2}, contributes to the $K_{\ell}^2 K_L^{-1}$ terms in \eqref{eq:intermediate_rough_lb}, and the last term to the relative $n_+^H$ term in \eqref{eq:intermediate_rough_lb}.

Collecting the inequalities \eqref{eq:Q3estim_rhofar},  \eqref{eq:int_estimate_rhofar} and \eqref{ine:Q2}, we deduce the lower bound in \eqref{eq:intermediate_rough_lb}.

By the simple algebraic equivalence 
\begin{equation}
\frac{1}{2} \rho_z^2 +\rho^2 - \rho \rho_z= \frac{1}{2}(\rho - \rho_z)^2 +\frac{1}{2}\rho^2,
\end{equation}
and using that the coefficients of the $K_{\ell}^2 K_L^{-1}$ in \eqref{eq:intermediate_rough_lb} can be bounded by 
\begin{equation}
C (\rho-\rho_z)^2 \widehat{g}(0)|\Lambda| + C \rho^2 \widehat{g}(0)|\Lambda|, 
\end{equation}
we get the bound
\begin{align}
\eqref{eq:intermediate_rough_lb} &\geq \Big(\frac{1}{2} \rho^2 |\Lambda|\widehat{g}(0) + \frac{1}{2} (\rho - \rho_z)^2|\Lambda| \widehat{g}(0) (1- C K_{\ell}^2 K_L^{-1}) \nonumber\\
&\qquad - C \rho^2 |\Lambda|\widehat{g}(0) K_{\ell}^2 K_L^{-1} - C\rho^2|\Lambda|\widehat{g}(0) \f\Big) \| \Phi(z) \|^2 -C \rho \widehat{g}(0) \langle n_+^H\rangle_{\Phi(z) } \nonumber\\
&\geq \Big(\frac{1}{2} \rho^2 |\Lambda|\widehat{g}(0) + \frac{1}{4} (\rho - \rho_z)^2|\Lambda| \widehat{g}(0) - C \rho^2 |\Lambda|\widehat{g}(0)K_{\ell}^2 K_L^{-1}\nonumber \\
&\qquad - C\rho^2 |\Lambda|\widehat{g}(0) \f \Big)\| \Phi(z) \|^2 - C \rho \widehat{g}(0) \langle n_+^H\rangle_{\Phi(z) } , 
\end{align}
and we can conclude using the assumption $|\rho - \rho_z| \geq \rho \varepsilon_+$, where $\varepsilon_+$ is chosen in order to dominate the $K_{\ell}^2 K_L^{-1}$terms and the error and to have that the second term in the previous expression positive and bigger than the Lee-Huang-Yang precision, to obtain the desired bound. 
\end{proof}

\section{A priori Bounds for the Number of Excited Bosons}
\label{app:C}
In this section we bound the number of excitations for states of suitably low energy. 

\begin{theorem}\label{thm:apriori_n+}
                         Assume the relations between the parameters in \ref{app:parameters} and that $\rho a^d$ is small enough. There exists a $C_B>0$ such that, if $\Psi \in L^2_{\text{sym}}(\Lambda^N)$ is a normalized state satisfying 
                        \begin{equation}\label{eq:assumption_lowE_psi}
                                                    \langle \mathcal{H}\rangle_{\Psi} \leq \frac{1}{2} \rho^2 \vert \Lambda \vert\widehat{g}(0)(1  + C_B \f),
                          \end{equation}
                          then there exists a $C>0$ such that      
                          \begin{equation}\label{eq:condensationestimate}
                                                    \langle n_+\rangle_{\Psi} \leq C \begin{dcases}
                                                                               C_B N  K_{\ell}^2   \widehat{g}(0), &d=2,\\
                                                                               C_B N K_{\ell}^2  \sqrt{\rho a^3}, &d=3. 
                                                    \end{dcases}
                          \end{equation}
                          \begin{equation}\label{bound_n_+H}
                          \langle n_+^{H}\rangle_{\Psi} \leq C \begin{dcases}
                                                    C_B N\, K_L^{-2} K_{\ell}^2   \widehat{g}(0) , &d=2,\\
                                                    C_B N\,K_L^{-2} K_{\ell}^2  \sqrt{\rho a^3}, &d=3. 
                          \end{dcases}
\end{equation}
\begin{equation}
\langle \mathcal Q_4^{\rm{ren}} \rangle_\Psi \leq C_B \rho^2  \vert \Lambda \vert \widehat{g}(0) \f. \label{eq:estimateprioriQ4}
\end{equation}
\end{theorem}

In order to prove the Theorem~\ref{thm:apriori_n+}, we need to prove a lower bound on $\mathcal{H}$ localizing on boxes $B$ with Gross-Pitaevskii length scale $\ell_{\text{GP}}\ll \ell$, where
\begin{equation}
\ell_{\text{GP}} := \rho^{-1/2} \widehat{g}(0)^{-1/2}.
\end{equation}
We introduce the small box centered at $u \in \Lambda$ to be 
\begin{equation}
	B_u =    u + \Big[-\frac{\ell_{\text{GP}}}{2} , \frac{\ell_{\text{GP}}}{2}  \Big]^d.
\end{equation}
The associated localization functions are
\begin{equation}\label{eq:Blocfunction}
	\chi_{B_u}(x) := \chi \left( \frac{x - u}{\ell_{\text{GP}}} \right),
\end{equation}
where $\chi \in C^{\infty}(\mathbb{R}^d)$, $0 \leq \chi$, $\mathrm{supp }\, \chi \subseteq B_{\frac{1}{2}}(0)$, $\|\chi\|_{L^2} = 1$. 
We emphasize that 
\begin{equation}
	\int_{\Lambda}\int_{B_u} |\chi_{B_u}|^2 \dd x \dd u = |\Lambda|.
\end{equation}
We also introduce the projectors on the condensate in the small boxes $P_{B_u}$ and their complements $Q_{B_u}$,
\begin{equation}
P_{B_u} := \frac{1}{|B_u|} |\one_{B_u}\rangle \langle \one_{B_u}|, \qquad Q_{B_u}  := \one_{B_u} - P_{B_u}.
\end{equation}
In order to construct the small box Hamiltonian, we introduce the localized potentials 
\begin{align}\label{eq:SF_3.5}
	v^{B}(x) &:= \frac{v(x)}{\chi*\chi(x/\ell_{\text{GP}})},  &w_{B_u}(x,y) := \chi_{B_u} (x) v^B(x-y)\chi_{B_u}(y), \\
	v_1^{B}(x) &:= \frac{g(x)}{\chi*\chi(x/\ell_{\text{GP}})},  &w_{1,B_u}(x,y) := \chi_{B_u} (x) v_1^B(x-y)\chi_{B_u}(y),  \\
	v_2^{B}(x) &:= \frac{g(x)(1+\omega(x))}{\chi*\chi(x/\ell_{\text{GP}})},  &w_{2,B_u}(x,y) := \chi_{B_u} (x) v_2^B(x-y)\chi_{B_u}(y),
\end{align} 
where we see that $w_{B}, w_{1,B}, w_{2,B}$ are localized versions of $v, g, (1+ \omega) g$, respectively.

For the kinetic energy, the localization to the small boxes is contained in the lemma below.

\begin{lemma}\label{lem:loc}
There exists a constant $b>0$ such that, for $s>0$ small enough, the periodic Laplacian on $\Lambda$ satisfies
\begin{equation}
- \Delta \geq \vert B \vert^{-1} \int_{\Lambda} \mathcal T_u \,\dd u + \frac{b}{\ell^2} Q_{\Lambda},
\end{equation}
where $Q_{\Lambda}$ is the projector outside the condensate of the box $\Lambda$, and where the new kinetic energy has the form
\begin{equation}
\mathcal T_u := Q_{B_u} \chi_{B_u} \Big( -\Delta_{\R^d} - s^{-2}\ell_{\rm{GP}}^{-2} \Big)_+ \chi_{B_u} Q_{B_u} + b \ell^{-2}_{\rm{GP}} Q_{B_u}.
\end{equation}
\end{lemma}

\begin{proof}
The proof can be found in  \cite[Lemma 3.3]{FourBEC}.
\end{proof}

Since we do not know how the particles distribute in the boxes, we introduce a chemical potential $\rho_{\mu}$. We will impose $\rho_{\mu} = \rho$ to be coherent with the original density. In this way we can define the grand canonical large box Hamiltonian, on the sector with $n$ bosons, as  
\begin{equation}
	\mathcal{H}_{\Lambda}(\rho_{\mu})_n := \sum_{j=1}^n \Big(-\Delta_j  - \rho_{\mu} \int_{\mathbb{R}^d}  g(x_j-y)\mathrm{d}y\Big) + \sum_{i<j}^n v(x_i-x_j).
\end{equation}
The small-box Hamiltonian $\mathcal{H}_B$ which acts on $\mathscr{F}_s(L^2(B_u))$ is 
\begin{equation}
	\mathcal{H}_{B_u}(\rho_{\mu})_n := \sum_{j=1}^n \Big(\mathcal{T}_{j,u}  - \rho_{\mu} \int_{\mathbb{R}^d}  w_{1,B_u}(x_j,y)\mathrm{d}y\Big) + \sum_{i<j}^n w_{B_u}(x_i,x_j).
\end{equation}

Joining Lemma~\ref{lem:loc} and a direct calculation for the potential, we obtain the relation between the last two Hamiltonians in the theorem below.

\begin{theorem}
	\begin{equation}\label{eq:lowerlambda_by_B}
		\mathcal{H}_{\Lambda}(\rho_{\mu})_n \geq \sum_{j=1}^n \frac{b}{\ell^2} Q_{\Lambda,j} + \frac{1}{\lvert B \rvert}\int_{\Lambda} \, \mathcal{H}_{B_u}(\rho_{\mu})_n \mathrm{d} u.
	\end{equation}
\end{theorem}
A lower bound for $\mathcal{H}_{B_u}$ gives a lower bound for $\mathcal{H}_{\Lambda}(\rho_{\mu})_n$ still conserving the contribution from the spectral gap. In the next proposition we give a lower bound for $\mathcal{H}_{B_u}$. The proof, that we omit, is identical to the one given in \cite{FourBEC} for the $3D$ case (see also \cite[Appendix B]{FS} and \cite[Appendix D]{2DLHY}).
\begin{proposition}
	Assume the conditions in \ref{app:parameters} are true, then there exists a constant $C_B>0$ such that, for sufficiently small values of $\rho_{\mu} a^d$,
	\begin{equation}
		\mathcal{H}_{B}(\rho_{\mu})_n \geq -\frac{1}{2}\rho_{\mu}^2 \lvert B \rvert\widehat{g}(0) - C_B \rho_{\mu}^2 \lvert B \rvert \widehat{g}(0) \lambda^{\rho_{\mu}}_d,  
	\end{equation}
	where $\lambda^{\rho_{\mu}}_d$ has the same expression as $\f$, with $\rho_{\mu}$ in place of $\rho$.
\end{proposition}
Plugging the result of this last proposition into \eqref{eq:lowerlambda_by_B}, and since all the $\mathcal{H}_{B_u}$ are unitarily equivalent, we get a lower bound for the large box Hamiltonian, contained in the next theorem.
\begin{theorem}
	\label{thm:small_largeHam_lowerbound}
	We have the following lower bound for the large box Hamiltonian
	\begin{align}
		\mathcal{H}_{\Lambda}(\rho_{\mu})_n \geq \frac{b}{2 \ell^2} n_+  -  \rho_{\mu}^2 |\Lambda| \widehat{g}(0)  \Big(\frac{1}{2}  + C_B \lambda_d^{\rho_{\mu}} \Big).\label{eq:roughlowerboundlarge}
	\end{align}
\end{theorem}

To lower bound the large box Hamiltonian by the spectral gap plus the energy contribution up to the Lee-Huang-Yang level, allows us to finally prove the bound on the number of excitations for states of low energy.
\begin{proof}[Proof of Theorem~\ref{thm:apriori_n+}] We only sketch the proof, details can be found in \cite[Appendix D]{2DLHY} and \cite[Appendix B]{FS}.
	Choosing $\rho_{\mu} = \rho$ we have  that the original large box Hamiltonian can be expressed, in relation to the grand canonical one, as 
	\begin{equation}
	\mathcal{H}_N = \mathcal{H}_{\Lambda}(\rho)_N + \rho\widehat{g}(0) N.
	\end{equation} Therefore, comparing the upper bound from the assumption \eqref{eq:assumption_lowE_psi} on $\Psi$ and the lower bound from Theorem~\ref{thm:small_largeHam_lowerbound}, we get
	\begin{equation}
		\frac{b}{2 \ell^2} \langle n_+\rangle_{\Psi}  + \frac{1}{2} \rho^2 |\Lambda| \widehat{g}(0)  -C_B\rho^2|\Lambda|\widehat{g}(0)\f 
		\leq \langle \mathcal{H}\rangle_{\Psi} \leq \frac{1}{2} \rho^2 |\Lambda| \widehat{g}(0) + C_B \rho^2 |\Lambda|\widehat{g}(0)\f, 
	\end{equation}
	which yields, for $n_+$, 
	\begin{equation}
		\frac{b}{2 \ell^2} \langle n_+\rangle_{\Psi} \leq 2 C_B \rho^2 |\Lambda| \widehat{g}(0) \f,
	\end{equation}
	giving the desired bound.
	
	The bound of $n_+^H$ follows from the one of $n_+$ and a lower bound on the Hamiltonian in the large box $\Lambda$, and we give a sketch of the proof below.
	
	We write the Laplacian in second quantization and on the $N$ boson space as
	\begin{equation}
	-\Delta = \sum_{k \in \Lambda^*} \tau_k a_k^{\dagger} a_k + b\frac{K_L^2}{\ell^2}n_+^H, 
	\end{equation}
where, for a $b<\frac{1}{100}$,
\begin{equation}
\tau_k := |k|^2 -b\one_{[K_L \ell^{-1}, +\infty)}(k) \frac{K_L^2}{\ell^2}, 
\end{equation} 
isolating, in this way, the spectral gap for high momenta. Thanks to this observation and Proposition~\ref{prop:Hsecondquant}, the Hamiltonian acting on the $N$ Fock space sector can be bounded as
\begin{align*}
\mathcal{H}_n &\geq \mathcal{K}_{\text{quad}}  + b\frac{K_L^2}{\ell^2}n_+^H + \frac{n_0(n_0-1)}{2 \vert \Lambda \vert} \big( \widehat g (0) + \widehat{g \omega}(0) \big)  \nonumber \\
& \quad + \mathcal Q_3^{\rm{ren}} + \mathcal Q_4^{\rm{ren}}- C n\widehat g(0)   \frac{ n_+}{\vert \Lambda \vert },
\end{align*}
where by $\mathcal{K}_{\text{quad}}$ we denoted the quadratic part of the Hamiltonian in $a_k^{\#}$:
\begin{equation}
\mathcal{K}_{\text{quad}} := \sum_{k \in \Lambda^*} \tau_k a_k^{\dagger} a_k + \frac{1}{2|\Lambda|} \sum_{k \in \Lambda^*} \widehat g_k (a_0^\dagger a_0^\dagger a_k a_{-k} + h.c.).
\end{equation}
Here we do not need to reach the Lee-Huang-Yang precision, therefore we do not have to work with soft pairs and the bound on $\mathcal{Q}_3^{\text{ren}}$ and $\mathcal{Q}_4^{\text{ren}}$ is easier. It is obtained by an application of a Cauchy-Schwarz inequality on $\mathcal{Q}_3^{\text{ren}}$ and estimating the missing terms to reconstruct $\mathcal{Q}_4^{\text{ren}}$ in a similar way as in \eqref{eq:PQPQ_n+estimate}:
\begin{equation}\label{n+H_Q3Q4}
\mathcal{Q}_3^{\text{ren}} + \frac{1}{2}\mathcal{Q}_4^{\text{ren}} \geq -C \frac{n_0}{|\Lambda|}n_+ \widehat{g}(0).
\end{equation} 
We introduce a new pair of creation and annihilation operators
\begin{equation}
b_k:= a_0^{\dagger} a_k, \qquad b^{\dagger}_k := a_0 a^{\dagger}_k,
\end{equation}
and adding and subtracting 
\begin{equation}
A_0 := \frac{\widehat{g}(0) }{2\lvert \Lambda \rvert }\sum_{k \in \Lambda^*}(b_k^{\dagger} b_k + b_{-k}^{\dagger} b_{-k}), 
\end{equation}
where $|A_0| \leq C N \widehat{g}(0) \frac{n_+}{|\Lambda|}$, we get
\begin{align*}
\mathcal{K}_{\text{quad}} + A_0 &\geq   \frac{1}{2\lvert \Lambda \rvert} \sum_{k \in \Lambda^*}\Big( \mathcal{A}_k (b_k^{\dagger} b_k + b_{-k}^{\dagger} b_{-k})+ \widehat{g}_k (b_k^{\dagger} b_{-k}^{\dagger} + b_{-k}b_k)  \Big) 
\end{align*}
where $\mathcal{A}_k := \frac{|\Lambda|}{(N+1)} \tau_k + \widehat{g}(0)$. By the standard Bogoliubov theory of diagonalization and recalling the definition of $G_d$ in \eqref{defG}, we bound the previous expression by the Bogoliubov integral
\begin{equation}
\mathcal{K}_{\text{quad}} + A_0 \geq  I(d)  -\frac{N(N+1)}{2\lvert \Lambda \rvert} \widehat{g\omega}(0),
\end{equation}
with
\begin{equation}
I(d) := -\frac{N}{2(2\pi)^d} \int_{\mathbb{R}^d} \Big(\mathcal{A}_k -\sqrt{\mathcal{A}_k^2- \widehat{g}_k^2} - \frac{N+1}{\lvert \Lambda \rvert}G_d(k) \Big) \mathrm{d}k. 
\end{equation}
We calculate the integral in a similar way as in Lemma~\ref{lem:Bog_int_parameters}, splitting into two regions for momenta higher or lower than $K_L\ell^{-1}$, obtaining, since $K_{\ell} \ll K_L$, that there exists a $C >0$, such that 
\begin{align}\label{n+H_bogint}
I(3) \geq - C \frac{N(N+1)}{\lvert \Lambda \rvert} \widehat{g}(0) \sqrt{\rho \widehat{g}(0)^3} \frac{K_L}{K_{\ell}},\qquad
I(2) \geq -  C \frac{N(N+1)}{\lvert \Lambda \rvert} \widehat{g}(0)^2,
\end{align}
Collecting the inequalities \eqref{n+H_bogint}, the bound on $A_0$ and \eqref{n+H_Q3Q4}, using the bound we obtained for $n_+$ and considering the quadratic form of the $N-$particle state $\Psi$ from the assumptions, we get the following lower bound for the Hamiltonian: 
\begin{equation}
\langle \mathcal{H} \rangle_{\Psi} \geq b\frac{ K_L^2}{ \ell^2}\langle n_+^H\rangle_{\Psi} + \frac{1}{2}\langle \mathcal{Q}_4^{\text{ren}}\rangle_{\Psi}+ \frac{1}{2} \rho N \widehat{g}(0)  \times \begin{dcases}
\Big(1 - C \sqrt{\rho a^3} \frac{K_{L}}{K_{\ell}}\Big), \quad &\text{for } d=3,\\
\Big(1 - C \widehat{g}(0) \Big), \quad &\text{for } d=2,
\end{dcases}
\end{equation}
which, together with the assumption \eqref{eq:assumption_lowE_psi} on $\Psi$, gives the bounds
\begin{align}
&\langle \mathcal{Q}_4^{\text{ren}}\rangle_{\Psi} \leq C \rho N \widehat{g}(0) \sqrt{\rho a^3}, \\
& b\frac{ K_L^2}{2 \ell^2} \langle n_+^H \rangle_{\Psi} \leq C \rho N \widehat{g}(0)\times \begin{dcases}
C \sqrt{\rho a^3} \frac{K_{L}}{K_{\ell}}, \quad &\text{for } d= 3,\\
C \widehat{g}(0), \quad &\text{for }d=2, 
\end{dcases}
\end{align}
from which the bounds on $n_+^H$ and $\mathcal{Q}_4^{\text{ren}}$ follow.
\end{proof}

\section{Parameters}\label{app:parameters}
In this appendix we list the parameters needed in the proof and the relations they have to satisfy. Finally, in \eqref{eq:para.choice} below we give a concrete choice satisfying those conditions.
Throughout all the paper, the following parameters are used 
\begin{equation}\label{eq:eps}
 \varepsilon_K, \varepsilon_{\rm{gap}} \ll 1 \ll \mathcal{M}, K_{\ell}, K_L, K_H,
\end{equation}
We use the notation $A \ll B$ to mean 
\begin{equation}
A \ll B \; \Leftrightarrow \; \begin{cases} A \leq  C (\rho a^3)^{\zeta} B, & \text{if } d=3, \\
A \leq  C \delta^{\zeta} B, & \text{if } d=2.
\end{cases}
\end{equation}
for a constant $C>0$ and a $\zeta >0$.

Recall that $K_L$ and $K_H$ define the sets of low and high momenta respectively. They must satisfy 
\begin{equation}\label{KL_relations}
	K_{\ell} \ll K_{\ell}^4 \ll 	 K_L \ll K_H.
\end{equation}
The chain of conditions
is important in many inequalities throughout all the paper.
$\mathcal M$ is the bound on $n_+^L$ that we allow our states to satisfy. Our localization result on $n_+^L$, Theorem~\ref{thm:excitationrestriction}, respectively requires in equation \eqref{eq:d1d2.01} and in equation \eqref{eq:M}
\begin{align}\label{eq:boundonM}
\mathcal M \gg  \ell \rho^{1/2} K_H^{1/2} \Vert v \Vert_1^{1/2},
\end{align}
and 
\begin{align}\label{eq:boundonM1}
	\mathcal M \gg  \ell^{2} \rho^{2}  \vert \Lambda \vert\widehat{g}(0)\lambda_{d}^{\rm{LHY}}.
\end{align}
The parameter $\mathcal{M}$ has  to be smaller than the total number of particles according to the following condition
\begin{equation} \label{eq:boundonM2}
\frac{\mathcal{M}}{N} \ll \Big(\frac{K_{\ell}}{K_L}\Big)^4\ll 1,
\end{equation}
where the last inequality follows from \eqref{KL_relations} using \eqref{eq:eps}.
The errors when localizing the 3Q terms in Proposition~\ref{prop:Q3loc} require the following condition
\begin{align}\label{M-KH_relation}
\frac{\mathcal{M}}{N}K_H^d \ll 1.
\end{align}
When dealing with the 3Q terms, we need a small fraction $\varepsilon_K \ll 1$ of $\mathcal K_H^{\rm{diag}}$ to control some errors. This coefficient needs to be large enough,
\begin{align} 
\varepsilon_K^2 &\gg \ell^{8-d} \rho^3 \widehat g(0)^4 K_H^{-8} K_L^d \mathcal M. \label{eq:rel_epsK_KM3}
\end{align}
Other errors from 3Q are controlled by $n_+^H$ using
\begin{equation}
K_\ell^2 \gg \ell^{4-d} \rho \widehat g(0)^2 K_H^{-2} K_L^d \mathcal M , \label{eq:rel_epsK_KM2} 
\end{equation}
or by a fraction $\varepsilon_{\rm{gap}} \ll 1$ of the spectral gap, which needs to satisfy
\begin{align}
\rho^3 \widehat g(0)^4 \ell^{8-d} K_H^{d-6} &\ll \varepsilon_{\rm{gap}}, \label{eq:epsgap1} \\
\varepsilon_K \ell^2 \rho \widehat g(0) + \mathcal E_d \ell^2 \rho +  K_\ell^2 K_H^{-1} &\ll \varepsilon_{\rm{gap}}, \label{eq:epsgap2}\\
\frac R \ell K_L \ell^2 \rho \widehat g(0) &\ll \varepsilon_{\rm{gap}}, \label{eq:epsgap3}
\end{align}
where $\mathcal E_d$ is the error from Lemma~\ref{lem:gomega.approx}.

We explain here how to get explicit choices of parameters, starting from any box $\Lambda$ satisfying
\begin{equation} \label{condition.on.Kl} K_\ell \ll \begin{cases}
	\delta^{-\frac {1}{26}}, \; &\text{if} \; d=2, \\
	(\rho a^3)^{-\frac 1 {28}}, \; &\text{if} \; d=3.
\end{cases}
\end{equation}
Given such a $K_\ell$, there exists an $\varepsilon \in (0,1)$ small enough such that
\begin{equation}
\begin{cases}
	K_\ell^{-26-19 \varepsilon} \gg \delta, \; &\text{if} \; d=2, \\
	K_{\ell}^{-28-16 \varepsilon} \gg \rho a^3, \; &\text{if} \; d=3.
\end{cases}
\end{equation}
Then, with the choice
\begin{equation}\label{eq:para.choice}
\begin{matrix}
	K_L = K_\ell^{4+2\varepsilon}, &
	K_H = K_\ell^{4+3\varepsilon}, &
	\frac{\mathcal M }{N} = K_\ell^{-12 -10\varepsilon}, \\
	\varepsilon_{\rm{gap}} = K_\ell^{-2},&  \varepsilon_K = K_\ell^{-18 +2d +(d-16)\varepsilon},  &
\end{matrix}
\end{equation}
all the conditions \eqref{KL_relations}, \eqref{eq:boundonM}, \eqref{eq:boundonM2}, \eqref{M-KH_relation}, \eqref{eq:rel_epsK_KM2}, \eqref{eq:rel_epsK_KM3}, \eqref{eq:epsgap1}, \eqref{eq:epsgap2}, \eqref{eq:epsgap3} are satisfied, for potentials satisfying
$	\| v \|_1 \leq C $ and $ \rho \widehat g(0) R^2  \leq K_\ell^{-9}$.

\bibliographystyle{hplain}
\bibliography{procrefs}

\end{document}